\documentclass[12pt, onecolumn]{IEEEtran}
\makeatletter
\def\ps@headings{%
\def\@oddhead{\mbox{}\scriptsize\rightmark \hfil \thepage}%
\def\@evenhead{\scriptsize\thepage \hfil \leftmark\mbox{}}%
\def\@oddfoot{}%
\def\@evenfoot{}}
\makeatother
\usepackage{setspace}
\usepackage[letterpaper]{geometry}
\usepackage{graphicx}
\usepackage{amsmath,amssymb}
\usepackage{amsmath}
\usepackage{multirow}
\usepackage{verbatim}
\usepackage{epstopdf}
\usepackage{epsfig}
\usepackage{cite}
\usepackage{color}
\usepackage{algorithm2e}
\usepackage{algorithmic}
\usepackage{subfigure}
\usepackage{lipsum}
\usepackage[framemethod=default]{mdframed}
\newmdenv[linecolor=black,backgroundcolor=white]{myframe}

\DeclareMathOperator{\argmax}{argmax}

\newcommand{\st}{\text{s.t.}}

\newtheorem{theorem}{Theorem}[section]
\newtheorem{remark}[theorem]{Remark}
\newtheorem{proposition}[theorem]{Proposition}
\newtheorem{lemma}[theorem]{Lemma}

\newtheorem{exam}{Example}
\newtheorem{proof}{Proof}

\ifodd 0
\newcommand{\rev}[1]{{\color{blue}#1}} 
\newcommand{\com}[1]{\textbf{\color{red}(COMMENT: #1)}} 
\newcommand{\clar}[1]{\textbf{\color{green}(NEED CLARIFICATION: #1)}}

\newcommand{\response}[1]{\textbf{\color{magenta}(RESPONSE: #1)}} 
\else
\newcommand{\rev}[1]{#1}
\newcommand{\com}[1]{}
\newcommand{\clar}[1]{}
\newcommand{\response}[1]{}
\fi
\newcommand{\RNum}[1]{\uppercase\expandafter{\romannumeral #1\relax}}

\begin{document}
\title{Revisiting Optimal Power Control: Its Dual Effect on SNR and Contention}
\author{
\IEEEauthorblockN{Yang Liu$^{1}$, Mingyan Liu$^{1}$ and Jing Deng$^{2}$}\\
\IEEEauthorblockA{
$^1$Electrical Engineering and Computer Science, University of Michigan, Ann Arbor\\
\{youngliu, mingyan\}@umich.edu\\
$^2$Department of Computer Science, University of North Carolina at Greensboro\\
jing.deng@uncg.edu
}
}

\maketitle
\nocite{*}

\begin{abstract}
\rev{In this paper we study a transmission power-tune/control problem in the context of 802.11 Wireless Local Area Networks (WLANs) with multiple (and possibly densely deployed) access points (APs). Previous studies on power control tend to focus on one aspect of the control, either its effect on transmission capacity (PHY layer effect) assuming simultaneous transmissions, or its effect on contention order (MAC layer effect) by maximizing spatial reuse.  We observe that power control has a {\em dual effect}:  it affects both spatial reuse and capacity of active transmission; moreover, maximizing the two separately is not always aligned in maximizing system throughput and can even point in opposite directions. 
In this paper we introduce an optimization formulation that takes into account this dual effect, by measuring the impact of transmit power on system performance from both PHY and MAC layers. 
We show that such an optimization problem is intractable and develop an analytical framework to construct simple yet efficient solutions.  Through numerical results, we observe clear benefits of this dual-effect model compared to solutions obtained that try to maximize spatial reuse and transmission capacity separately. 
This problem does not invoke cross-layer design, as the only degree of freedom in design resides with transmission power. It however highlights the complexity in tuning certain design parameters, as the change may manifest itself differently at different layers which may be at odds.} 

We further form a game theoretical framework and investigate above power-tune problem in a strategic network. We show that with decentralized and strategic users, the Nash Equilibrium (N.E.) of the corresponding game is in-efficient and thereafter we propose a punishment based mechanism to enforce users to adopt the social optimal strategy profile under both perfect and imperfect sensing environments.
\end{abstract}

\section{Introduction}
Power-tune has emerged as an important issue in an IEEE 802.11 WLAN network of multiple interacting users (Access Points, or APs). Earlier classical results with focus on power-tune may be classified into the following two independent approaches. 

The first relies on a PHY-layer framework in interference-bounded networks, i.e., the optimal power-tune problem is defined with respect to the Signal-Noise-Ratio (SNR) of each AP or the entire network. Within this framework, each AP's transmission power has two contradicting roles: The first is that a higher transmission power will improve the noise resistance capability for its own communication and thus potentially the network capacity. The other role is the unavoidable interaction with other APs. A higher transmission power will contribute higher noise/interference to other APs using the same channel (we assume Orthogonal Frequency Division Multiplexing, OFDM, at PHY layer and thus we will not consider intra-channel interference).  Many results have been established in this framework with different techniques focusing on either centralized or distributed solutions.  

The second class of results stems from MAC layer techniques by trying to reduce the level of contention  within a network, or improving \emph{spatial reuse order}, as more generally referred to. 
Specifically, when users fall into each other's audible range, transmission back-off under CSMA/CA is triggered to resolve contention and enable sharing. Therefore, decreasing users' transmission range helps improve spatial reuse of a given channel.  It follows that they are often modeled as congestion games or other similar graph problems 
(more in Section \ref{sec:related}). 

Even though conceptually both frameworks aim at optimizing system performance, e.g., overall throughput, the technical objectives and thus the net impact under the two are clearly not always aligned, and in fact can be quite different and even point in opposite directions.  To illustrate, consider maximizing users' 
achievable throughput or capacity without considering the induced spatial contention relationship; the resulting power-tune can create areas of very high contention order. 
Thus, even though a user's (or the network's) transmission capacity/rate maybe maximized on a per transmission basis, significant amount of air time may be spent in the back-off process instead of data transmission, leading to wasted spectrum resources.
The opposite may also be true.  If we simply control the contention topology of the network, the transmission power settings may be such that users do not have sufficient noise resistance capability and thus fall short of the theoretically achievable capacity.  In this case, even though we may have successfully reduced the contention and saved a lot of air time, the quality of active transmissions (or on a per-transmission basis) may be low.  

In short, reducing transmission power has a dual effect on the MAC and PHY layers: it can help increase spatial reuse order under CSMA/CA, but can at the same time decrease noise resistance capacity and therefore the transmission capacity. 
A desirable solution should thus take both effects into consideration in determining the optimal power control.  This is strictly speaking not a cross-layer problem, as the only degree of freedom in design resides with transmission power, i.e., there is no joint design or feedback between different layers.  This problem simply highlights the complexity in tuning certain design parameters, as the change may manifest itself differently at different layers which could be self-defeating as illustrated above.  

In this study we approach this problem by introducing a performance measure (or utility function) based on the power-tune impact on both PHY and MAC layers simultaneously. An interesting technical aspect of this formulation is the combination of both continuous (SNR and PHY) and discrete (MAC or graph-based) elements in a single optimization problem.  Not surprisingly, this leads to  
an intractable optimization problem, 
whose properties and structures we then investigate to help construct good and efficient solution techniques. 
Extensive simulation is conducted to verify the effectiveness and performance of our solution approach. 
An equally important aspect of our study, besides solving the above optimization problem, is to obtain insight into how the resulting power-tune differs from the two approaches outlined earlier, each focusing on the effect on a single layer, respectively. 

Our formulation and solution are given within a centralized framework.  A natural next step is to examine distributed implementations of the solution, and similar optimization problems when users are strategic.  These remain interesting directions of future research, but are out of the scope of the present paper. 

The remainder of the paper is organized as follows. Section \ref{sec:model} gives the system model and problem formulation, while Section \ref{sec:central} characterizes the optimal solution. Section \ref{sec:simulate} provides extensive numerical results to evaluate our approach. 
Section \ref{sec:related} presents a literature review most relevant to the present study and Section \ref{sec:conclusion} concludes the paper.


\section{System model}\label{sec:model}

\subsection{Preliminaries}

Consider a WLAN network with $N$ APs denoted by the set $\Omega = \{1,2,...,N\}$. Each AP is associated with a number of stations with whom it communicates. 
Denote an AP's transmission power space (i.e., the set of power levels it may employ) by $\mathcal Q_i, i \in \Omega$. Different from many prior works, here we do {\em not} assume any finiteness of $\mathcal Q_i$; instead,  we will show that the finiteness of the {\em optimal} power profile follows naturally from our formulation.  We will assume $\mathcal Q_i, i\in \Omega$ are all closed and use $\overline{\mathcal P}_i$ and $\underline{\mathcal P}_i$ to denote the maximum and minimum value in $\mathcal Q_i$, respectively. 
The transmission power profile of all users is denoted by $\mathcal P = [\mathcal P_1,\mathcal P_2,...,\mathcal P_N]$. 

Each AP also has a certain attempt rate for channel access under IEEE 802.11, and these are denoted by the vector $\mathbf{p}:=[p_1,p_2,...,p_N]$, also referred to as the attempt rate profile. Channel gain (or path loss) from user $i$ to $j$ is denoted by $h_{ij}$. We will assume $h_{ij}, i,j \in \Omega$ stay unchanged during a single transmission; alternatively, we may view $h_{ij}$ as the expectation of channel dynamics. $\mathcal N_0$ denotes the average noise level, and $\mathcal P^i_{cs}$ denotes the carrier sensing (CS) threshold of the $i-$th AP.

\rev{For the rest of the paper we will use the terms {\em AP} and {\em user} interchangeably.} 

\subsection{Contention domain}
Due to the fact that many hardware/circuits put a requirement on CS signal's strength, some CS signals cannot be correctly decoded and the corresponding back-off actions will not be triggered; only those with strength higher than the CS threshold can be correctly identified. We thus define two notions of a {\em contention domain} for user/AP $i$. The first one $\Delta^r_i$, the receive contention domain, is the set of users/APs whose CS signals can be correctly decoded by user $i$; while the other $\Delta^t_i$, the transmit contention domain, is the set of users who can decode user $i$'s CS signals correctly. Mathematically we have
\begin{align}
\Delta^r_i = \{j: \mathcal P_j\cdot h_{ji} \geq \mathcal P^i_{cs}, j \in \Omega \backslash \{i\}\}
\label{eq-Delta-i}
\end{align}
\begin{align}
\Delta^t_i = \{j: \mathcal P_i\cdot h_{ij} \geq \mathcal P^j_{cs}, j \in \Omega \backslash \{i\}\}
\label{eq-Delta-i}
\end{align}
By definition, contention domain is closely related to \emph{spatial reuse}. With a larger contention domain, the degree of spatial reuse is potentially smaller around that user. 
Define $n_i(\mathbf{\mathcal P})$ to be the number of competing users of user $i$ under power profile $\mathbf{\mathcal P}$; i.e., $n_i(\mathbf{\mathcal P}):=|\Delta^r_i|$.  This will also be referred to as the contention order for user $i$.

\subsection{Neighborhood reaching threshold}

Consider AP $i$ and the maximum (resp. minimum) transmission power it can use without reaching (resp. so that it can still reach) another AP $k$, expressed as follows (assumed to exist): 
\begin{align}
{\mathcal P}^{-}_{ik}:&= \max\{\mathcal P_i: \mathcal P_i \cdot h_{ik} < \mathcal P^k_{cs}, \mathcal P_i \in \mathcal Q_i\} \\
{\mathcal P}^{+}_{ik}:&= \min\{\mathcal P_i: \mathcal P_i \cdot h_{ik} \geq \mathcal P^k_{cs}, \mathcal P_i \in \mathcal Q_i\}~. 
\end{align} 
\rev{To make it complete for $k=i$ we have
\begin{align}
{\mathcal P}^{-}_{ii}:&=  \underline{\mathcal P}_i, {\mathcal P}^{+}_{ii}:=\overline{\mathcal P}_i 
\end{align}
} 
Denote the set of these neighbors reaching thresholds for AP $i$ as
\begin{align}
\tilde{\mathcal Q}_i = \cup_{k\in \Omega} \{{\mathcal P}^{-}_{ik},{\mathcal P}^{+}_{ik} \}.  
\end{align}
%
Denote the neighbor reaching profile space for the whole network as 
\begin{align}
\tilde{\mathcal Q} := \tilde{\mathcal Q}_1 \times \tilde{\mathcal Q}_2 \times ... \times \tilde{\mathcal Q}_N ~. 
\end{align}
Since we are considering a finite size network (i.e., the number of APs, $N$, is a finite positive integer), this profile space is finite, i.e., $|\tilde{\mathcal Q}_i| < \infty, \forall i \in \Omega$, and  consequently $|\tilde{\mathcal Q}| < \infty$. 

\subsection{A performance measure/utility function}

From AP $i$'s point of view, its transmission power setting has the following implications: 
\begin{itemize}
\item[\RNum{1}.] Higher transmission power will increase AP $i$'s received SNR ($\text{SNR}_i$) by its associated stations. 
\item[\RNum{2}.] Higher transmission power will cause higher interference to APs outside $\Delta^r_i \cup \Delta^t_i$.
\item[\RNum{3}.] Higher transmission power will add $i$ to some other AP $j$'s contention domain.
\end{itemize}

As a result, AP $i$'s perceived performance, or utility $\mathcal U_i$, as a function of the transmission power profile ${\mathcal P}$ and attempt rate profile $\mathbf{p}$ across all APs, may be captured by the following expression: 
\begin{align}
\mathcal U_i(\mathbf{\mathcal P}, \mathbf{p}) = \mathcal S_i(\mathbf{\mathcal P}, \mathbf{p})\cdot \mathcal C_i(\mathbf{\mathcal P}, \mathbf{p}), ~~~ i \in \Omega, 
\end{align}
where $\mathcal S_i$ is the ``sharing'' function representing AP $i$'s share of channel access under CSMA/CA-type of collision avoidance mechanism, and $\mathcal C_i$ is the ``capacity'' function representing the rate/quality of active transmissions under $\mathbf{\mathcal P}$ and $\mathbf{p}$. 


Under 802.11, we can approximate $\mathcal S_i$ using the probability of successful channel reservation given by 
$\prod_{j \in \Delta^r_i} (1-p_j)p_i$, 
where the dependence on the transmission power profile ${\mathcal P}$ is implicit through the contention domain $\Delta^r_i = \Delta^r_i({\mathcal P})$. 
Assuming a \emph{fair} WLAN network with $0 < p_1 = p_2 =...=p_N = p_c < 1$, we then have the following form of the sharing function: 
\begin{align}
\mathcal S_i(\mathbf{\mathcal P, \mathbf{p}}) = \prod_{j \in \Delta^r_i} (1-p_j)p_i = (1-p_c)^{n_i(\mathcal P)}p_c, 
\end{align}
which approximates the air time share of AP $i$ within its contention domain.

\rev{$\mathcal C_i$ is intended to capture the rate or capacity of active transmission attained by AP $i$.  To make this concrete, we will focus on the downlink capacity from the AP to its associated stations.  As the stations' locations can change dynamically and are often unknown, we will not explicitly model this level of detail and simply assume that the stations are sufficiently close to their associated AP.  Consequently, their capacity may be approximated using the transmit power by the AP (rather than the received power at a station) and the interference at the AP (rather than at a station)},    
in the standard Shannon formula: (similar formulation can also be found in \cite{DBLP:conf/globecom/TanPC05}): 
\begin{align}
\mathcal C_i(\mathbf{\mathcal P}, \mathbf{p}) = \log (1 + \frac{\mathcal P_i}{\mathcal N_0 + \sum_{j \in \bar{\Delta}^r_i \cap \bar{\Delta}^t_i} \mathcal S_j(\mathbf{\mathcal P, \mathbf{p}})\cdot\mathcal P_j\cdot h_{ji}})~,  
\label{eq-Ci-P-p}
\end{align}
where $\bar{\Delta}^{r}_i$ ($\bar{\Delta}^{t}_i$) denotes the complement of AP $i$'s contention domains, i.e., $\bar{\Delta}^{r}_i = \Omega_{-i} \backslash \Delta^{r}_i$, reflecting the fact that the interference comes primarily from APs outside the contention domain as a result of the back-off mechanism of IEEE 802.11 collision avoidance.
\section{The optimal power-tune problem and its characterization}\label{sec:central}
In this section, we formally define our optimization problem and do so in comparison with its single-layer counterparts, i.e., that aim at only the PHY or MAC layer effect, respectively.  We then characterize its features using a two-user example. 

\subsection{Considering only PHY layer effects}
When we limit our attention to PHY layer effects of power control, typically no contention is considered and parallel transmissions are implicitly assumed. Therefore each single user's transmissions will contribute to other users noise/interference. Problems along this line have been well investigated in the literature, see e.g., \cite{DBLP:conf/globecom/TanPC05, Chiang:2008:PCW:1454705.1454706}. Specifically these assumptions result in the following optimization problem
$$
\begin{array}{ccll}
\textbf{(PPHY)}&\max & \sum_{i \in \Omega} \log (1+\frac{\mathcal P_i}{\mathcal N_0 + \sum_{j \neq i} \mathcal P_{j}\cdot h_{ji}})\\
&\st& \mathcal P_{j} \in \mathcal Q_i, \forall i \in \Omega.\end{array}
$$
This rate maximization problem is in general hard to solve. Previous work has focused on different approximation techniques, see e.g., \cite{DBLP:conf/globecom/TanPC05}.  In order to obtain comparable results in order to compare this with our optimization formulation, in our numerical experiments (Section \ref{sec:simulate}) we shall use 
the following simple approximation 
\begin{align}
&\log (1+\frac{\mathcal P_i}{\mathcal N_0 + \sum_{j \neq i} \mathcal P_{j}\cdot h_{ji}}) \approx \log(\frac{\mathcal P_i}{\mathcal N_0})-\frac{\sum_{j \neq i} \mathcal P_{j}\cdot h_{ji}}{\mathcal N_0}
\end{align}
Since both terms $\log(\frac{\mathcal P_i}{\mathcal N_0})$ and $-\frac{\sum_{j \neq i} \mathcal P_{j}\cdot h_{ji}}{\mathcal N_0}$ are concave, we now have an approximate/relaxed optimization problem which is convex:  
$$
\begin{array}{ccll}
\textbf{(R-PPHY)}&\max & \sum_{i \in \Omega} \{\log(\frac{\mathcal P_i}{\mathcal N_0})-\frac{\sum_{j \neq i} \mathcal P_{j}\cdot h_{ji}}{\mathcal N_0}\}\\
&\st& \mathcal P_{j} \in \mathcal Q_i, \forall i \in \Omega.\end{array}
$$
This problem can be efficiently solved using classical convex optimization techniques (assuming all $\mathcal Q_i$ are convex or piece-wise convex).  These are used in our numerical results provided later, and the algorithmic details are omitted for brevity. 
 
\subsection{Considering only MAC layer effects} 

We next limit our attention to MAC layer effects of power control, in which case the objective is typically to minimize the sum of 
contention orders over all users, given as follows: 
\begin{align}
\min \sum_{i \in \Omega} n_i (\mathcal P)~, 
\end{align}
see, e.g., a similar objective used in \cite{DBLP:conf/infocom/WanCDWY12}. 
However, without any constraint on $\mathcal P$, the above minimization could lead to somewhat pathologic solutions, i.e., with very low power, we can obtain  
$ \sum_{i \in \Omega} n_i (\mathcal P) \approx 0$ thereby achieving the minimum value. However, with very low transmission powers and relatively constant background noise, each AP's SNR is significantly impacted leading to poor transmission performance. 
Consequently, in order to make this model comparable with others considered in our numerical experiments, we will instead consider a similar optimization problem with an SNR constraint.  Specifically, we will try to  minimize the total contention order within the network, subject to a minimum requirement on each AP's SNR, as shown below: 
$$
\begin{array}{ccll}
\textbf{(PMAC)}&\min & \sum_{i \in \Omega} n_i(\mathcal P)\\
&\st&  \frac{\mathcal P_i}{\mathcal N_0 + \sum_{j \in \bar{\Delta}^r_i \cap \bar{\Delta}^t_i}S_j(\mathbf{\mathcal P, \mathbf{p}}) \cdot \mathcal P_j \cdot h_{ji}} \geq \text{SNR}_0, \\
&&\forall i \in \Omega.\end{array}
$$
Here we use $\text{SNR}_0$ to denote some baseline SNR that needs to be met at each AP's transmission. 

The above problem has a mixture of a combinatorial term and continuous term in the following sense: while the SNR is continuous w.r.t. setting transmission power $\mathcal P_i$, the contention domains $\Delta_i$s are discrete. Thus the problem is hard to solve in general. We thus consider a relaxation of the above problem. Since we have
\begin{align}
\mathcal P_j \cdot h_{ji} < \mathcal P^i_{cs}, \forall j \in \bar{\Delta}^r_i \cap \bar{\Delta}^t_i~, 
\end{align}
the inequality below holds immediately
\begin{align}
\sum_{j \in \bar{\Delta}^r_i \cap \bar{\Delta}^t_i} S_j(\mathbf{\mathcal P, \mathbf{p}})\cdot\mathcal P_j \cdot h_{ji} < \sum_{j \in \bar{\Delta}^r_i \cap \bar{\Delta}^t_i} S_j(\mathbf{\mathcal P, \mathbf{p}})\cdot \mathcal P^i_{cs}
\end{align}
Moreover, we have 
\begin{align}
\frac{\mathcal P_i}{\mathcal N_0 + \sum_{j \in \bar{\Delta}^r_i \cap \bar{\Delta}^t_i} \mathcal P_j \cdot h_{ji}} > \frac{\mathcal P_i}{\mathcal N_0 +\sum_{j \in \bar{\Delta}^r_i \cap \bar{\Delta}^t_i} S_j(\mathbf{\mathcal P, \mathbf{p}})\cdot \mathcal P^i_{cs}}. 
\end{align}
Thus, we have the following relaxed problem that are solvable:
$$
\begin{array}{ccll}
\textbf{(R-PMAC)}&\min & \sum_{i \in \Omega} n_i(\mathcal P)\\
&\st&\frac{\mathcal P_i}{\mathcal N_0 +\sum_{j \in \bar{\Delta}^r_i \cap \bar{\Delta}^t_i} S_j(\mathbf{\mathcal P, \mathbf{p}})\cdot \mathcal P^i_{cs}}\geq \text{SNR}_0, \forall i \in \Omega.
\end{array}
$$

\begin{theorem}
In solving \textbf{R-PMAC} there is no loss of optimality to limit our attention to the space $\tilde{\mathcal Q}$, i.e., if an optimal solution exists, we can always find an optimal solution within the space $\tilde{\mathcal Q}$. 
\end{theorem}
\begin{proof}
Suppose there is an optimal power profile with some element, say $\mathcal P^*_j \notin \tilde{\mathcal Q}_j$. Then consider relaxing/increasing $\mathcal P^*_j$ to the nearest ${\mathcal P}^{-}_{jk}, \forall k\in \Omega$. Note that such a change would not modify the contention topology and thus all $n_i$ values remain the same, without violating the corresponding SNR constraint.  Thus we have found an optimal solution within $\tilde{\mathcal Q}$. 
\end{proof}
\begin{remark}
This theorem allows us to focus on a finite set of solutions instead of all possible solutions.
\end{remark}

\subsection{Considering dual effects}

We now formalize the optimal power-tune problem outlined earlier that takes into account both the PHY and MAC layer effects. Specifically, we will seek centralized solutions to the following social welfare maximization problem: 
$$
\begin{array}{ccll}
\textbf{(P)}&\max & \mathcal U = \sum_{i \in \Omega} \mathcal U_i(\mathbf{\mathcal P}, \mathbf{p})\\
&\st& \mathcal P_i \in \mathcal Q_i, i \in \Omega. \end{array}
$$
As the power profile space $ \mathcal Q := \mathcal  Q_1 \times \mathcal  Q_2 \times ... \times \mathcal Q_N$ is potentially infinite, and $\mathcal U$ is in general a non-convex and non-differentiable function w.r.t. $\mathcal P$, the optimization problem is NP-hard.  \rev{To illustrate, Fig. \ref{func_spe} shows three examples of the sum utility as a function of the power $\mathcal P_i$ of AP $i$, under different parameter settings. 
 As can be seen, there lack properties commonly used to derive efficient solution techniques (e.g., differentiability,  convexity).} 
\rev{There are, however, some interesting features such as the prominent step shape shown in the result.  This observation motivated key results in the subsequent subsections.}
\begin{figure}[h!]
\centering
\includegraphics[width=0.5\textwidth,height=0.25\textwidth]{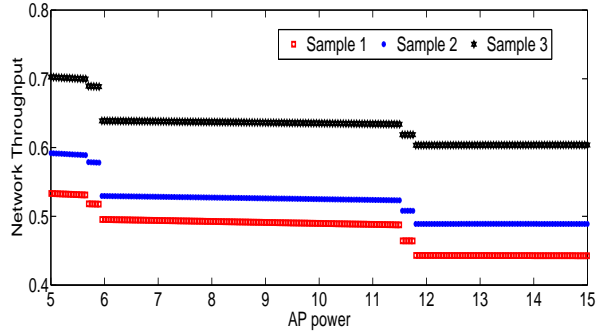}
\caption{Examples of $\mathcal U$ w.r.t. a specific $\mathcal P_i$}\label{func_spe}
\end{figure}

\subsection{An illustrative example: two-user case with centralized solution}
Consider the special case with $N=2$ and $p_c=\frac{1}{2}$, i.e., a system of two users attempting to access the channel each with probability 1/2.  To gain some insight into the solution to the corresponding optimization problem (P), we discuss the following four cases. 

\textbf{CASE 1.}  $\mathcal{P}_1 \leq \mathcal{P}^{-}_{12},\mathcal{P}_2 \leq \mathcal{P}^{-}_{21}$

In this case neither user is in the other's CS range and thus the objective function becomes
\begin{align}
\log(1+\frac{\mathcal P_1}{\mathcal N_0+\frac{1}{4}\cdot \mathcal P_2\cdot h_{21}})+\log(1+\frac{\mathcal P_2}{\mathcal N_0+\frac{1}{4}\cdot\mathcal P_1\cdot h_{12}})
\end{align} 
We first fix $\mathcal P_2$ and analyze the above function w.r.t. $\mathcal P_1$.  Basic algebraic calculation shows that the above objective function is equivalent to (noting $h_{12} = h_{21}$ and denote $h:=\frac{1}{4}\cdot h_{12}=\frac{1}{4}\cdot h_{21}$)
\begin{eqnarray*}
\log \left\{1+\frac{1}{\mathcal N_0+\mathcal P_2\cdot h}(\mathcal P_1+\mathcal P_2/h+\frac{\mathcal N_0+\mathcal P_2\cdot h-\mathcal P_2\cdot \mathcal N_0/h}{\mathcal N_0+\mathcal P_1\cdot h}) \right\}
\end{eqnarray*}
Due to the monotonicity of the $\log$ function, the optimization problem is equivalent to maximizing
\begin{align}\label{eqn:1}
\mathcal P_1+\mathcal P_2/h+\frac{\mathcal N_0+\mathcal P_2\cdot h-\mathcal P_2\cdot \mathcal N_0/h}{\mathcal N_0+\mathcal P_1\cdot h}~. 
\end{align}
If $\mathcal N_0+\mathcal P_2\cdot h-\mathcal P_2\cdot \mathcal N_0/h < 0$, then Eqn. (\ref{eqn:1}) is strictly increasing w.r.t. $\mathcal P_1$ and thus $\mathcal P^*_1 = \mathcal{P}^{-}_{12}$. 
If $\mathcal N_0+\mathcal P_2\cdot h-\mathcal P_2\cdot \mathcal N_0/h \geq 0$, then 
Eqn. (\ref{eqn:1}) is convex and the maximizer is an end point, i.e., $\mathcal P^*_1 \in \{\underline{\mathcal P}_1,  \mathcal{P}^{-}_{12}\}$. By symmetry we have similar results for user 2.

\textbf{CASE 2.} $\mathcal{P}_1 \geq \mathcal{P}^{+}_{12},\mathcal{P}_2 \leq \mathcal{P}^{-}_{21}$

\rev{In this case user 1 is in user 2's receive contention domain, but not the other way round.} 
It follows that the objective function in this case reduces to the following: 
\begin{align}
\log(1+\frac{\mathcal P_1}{\mathcal N_0})+\frac{1}{2}\log(1+\frac{\mathcal P_2}{\mathcal N_0})
\end{align}
Obviously this function is increasing w.r.t. $\mathcal P_1$ and $\mathcal P_2$ and therefore we have $\mathcal P^*_1 = \overline{\mathcal P}_1,\mathcal P^*_2 = \mathcal{P}^{-}_{21}$.  

\textbf{CASE 3.} $\mathcal{P}_1 \leq \mathcal{P}^{-}_{12},\mathcal{P}_2 \geq \mathcal{P}^{+}_{21}$

This is essentially the same as \textbf{CASE 2} and thus omitted.  

\textbf{CASE 4.} $\mathcal{P}_1 \geq \mathcal{P}^{+}_{12},\mathcal{P}_2 \geq \mathcal{P}^{+}_{21}$

In this case both users are in each other's contention domain, and the resulting objective function is 
\begin{eqnarray}
\frac{1}{2}\log(1+\frac{\mathcal P_1}{\mathcal N_0})+\frac{1}{2}\log(1+\frac{\mathcal P_2}{\mathcal N_0}). 
\end{eqnarray}
It is then straightforward to show that $\mathcal P^*_1 = \overline{\mathcal P}_1$ and $\mathcal P^*_2 = \overline{\mathcal P}_2$. 

\rev{The above example aims at conveying the intuition that the optimal power levels are likely to be at the neighborhood reaching thresholds or the maximum power limit.  The next section  shows that we can indeed restrict our attention to a finite set of these thresholds in search of an optimal solution.} 


\section{Solution Approach : A centralized View} 

\subsection{A Lower bound problem}
Recall that for $j \in \bar{\Delta}^r_i$, we have $\mathcal P_j \cdot h_{ji} < \mathcal P^i_{cs}$, i.e., the received signal strength is below the CS threshold.  Thus 
\begin{align}
\log (1&+\frac{\mathcal P_i}{\mathcal N_0 + \sum_{j \in \bar{\Delta}^r_i \cap \bar{\Delta}^t_i}S_j(\mathbf{\mathcal P, \mathbf{p}}) \cdot \mathcal P_j \cdot h_{ji}}) \nonumber \\
&> \log (1+\frac{\mathcal P_i}{\mathcal N_0 + \sum_{j \in \bar{\Delta}^r_i \cap \bar{\Delta}^t_i}S_j(\mathbf{\mathcal P, \mathbf{p}}) \cdot \mathcal P^i_{cs}}) ~. 
\end{align}
We use this relationship to form the following lower bound problem. 
$$
\begin{array}{ccll}
\textbf{(PL)}&\max &  \mathcal U_L= \sum_{i \in \Omega} (1-p_c)^{n_i(\mathcal P)}p_c\cdot \\
&~~&~~\log (1+\frac{\mathcal P_i}{\mathcal N_0 +  \sum_{j \in \bar{\Delta}^r_i \cap \bar{\Delta}^t_i}S_j(\mathbf{\mathcal P, \mathbf{p}}) \cdot \mathcal P^i_{cs}})\\
&\st& \mathcal P_i \in \mathcal Q_i, \forall i \in \Omega. \end{array}
$$
\begin{lemma}
For an optimal solution $\mathcal P^*$ to \textbf{(PL)}, we have
$\mathcal P^* \in \tilde{\mathcal Q}$.  That is, there is no loss of optimality is restricting the solution space to 
 $\tilde{\mathcal Q}$ in searching for an optimal solution. 
\end{lemma}
\begin{proof} 
For AP $i$, suppose there exists a $\mathcal P^*_i \not\in \tilde{{\mathcal Q}}_i$.  This  means one of the following must be true: (1) $\mathcal P^+_{ik} < \mathcal P^*_i < \mathcal P^-_{ij}$ for some $(k,j)$, (2) $\mathcal P^*_i < \mathcal P^-_{ij}$ for all $j$, and (3) $\mathcal P^+_{ik} < \mathcal P^*_i$ for all $k$. For the fist two cases, if we increase $\mathcal P_i$ from $\mathcal P^*_i$ to $\mathcal P^-_{ij}$, the resulting contention topology remain unchanged, i.e., the terms $\mathcal S_j, n_j(\mathcal P), \forall j \in \Omega$ stay unchanged, but $\mathcal P_i$ is now bigger and $\mathcal U_L$ is strictly increasing in $\mathcal P_i$.  This contradicts the optimality of $\mathbf {\mathcal P^*}$, so the first two cases cannot be true. If it's the third case and $\mathcal P^+_{ik} < \mathcal P^*_i$ for all $k$, then increasing $\mathcal P_i$ from $\mathcal P^*_i$ to $\overline{\mathcal P}_i$ results in the same argument as above, so (3) also cannot be true, completing the proof. 
\end{proof}
\begin{remark}
When $\mathcal P^i_{cs}$ is small or when the network is dense and large (i.e., with a large $N$), the lower bound problem will provide a good approximation for the original problem. 
\end{remark}

\subsection{An Upper bound problem}
We similarly form an upper bound to the original objective function: 
\begin{align}
\log (1&+\frac{\mathcal P_i}{\mathcal N_0 + \sum_{j \in \bar{\Delta}^r_i \cap \bar{\Delta}^t_i}S_j(\mathbf{\mathcal P, \mathbf{p}}) \cdot \mathcal P_j \cdot h_{ji}})  \nonumber \\
&\leq \frac{\mathcal P_i}{\mathcal N_0 + \sum_{j \in \bar{\Delta}^r_i \cap \bar{\Delta}^t_i}S_j(\mathbf{\mathcal P, \mathbf{p}}) \cdot \mathcal P_j \cdot h_{ji}} \cdot \log 2, 
\end{align}
and we have the following upper bound problem 
$$
\begin{array}{ccll}
\textbf{(PU)}&\max & \mathcal U_U=\sum_{i \in \Omega} (1-p_c)^{n_i(\mathcal P)}p_c\nonumber \\
&&\cdot \frac{\mathcal P_i}{\mathcal N_0 + \sum_{j \in \bar{\Delta}^r_i \cap \bar{\Delta}^t_i} S_j(\mathbf{\mathcal P, \mathbf{p}})\cdot \mathcal P_j \cdot h_{ji}}\\
&\st& \mathcal P_i \in \mathcal Q_i, \forall i \in \Omega.\end{array}
$$
\begin{lemma}\label{lem:convex}
$\mathcal U_U$ is piece-wise convex w.r.t. each $\mathcal P_i, i\in \Omega$. 
\end{lemma}
\begin{proof} 
Consider $\mathcal P_i$ and fix the transmission power of all other APs. 
%
Suppose $\mathcal P_i \in [\mathcal P^{+}_{ik}, \mathcal P^{-}_{ij}]$ for some $k, j$.  Within this range, the contention topology remains the same, i.e., $(1-p_c)^{n_i(\mathcal P)}p_c$ is a constant for any value $\mathcal P_i$ takes within this interval.  
Next consider the second term in $\mathcal U_U$. $\mathcal P_i$ appears in 
this term in two forms: one as $\frac{\mathcal P_i}{\mathcal N_0 +  \sum_{l \in \bar{\Delta}^r_i \cap \bar{\Delta}^t_i}S_l(\mathbf{\mathcal P, \mathbf{p}})\cdot \mathcal P_l\cdot h_{li}}$ which is convex  w.r.t. $\mathcal P_i$, and the other as $\frac{\mathcal P_l}{\mathcal N_0 +  \sum_{m\in \bar{\Delta}^r_l \cap \bar{\Delta}^t_l}S_m(\mathbf{\mathcal P, \mathbf{p}})\cdot \mathcal P_m\cdot h_{ml}}$ for some $l$ such that $i \in \bar{\Delta}^r_l \cap \bar{\Delta}^t_l$, in the form of $S_i(\mathbf{\mathcal P, \mathbf{p}})\cdot \mathcal P_i\cdot h_{il}$; 
these terms are also convex w.r.t. $\mathcal P_i$.  As the sum of convex functions is convex, we have established the convexity of $\mathcal U_{U}$. 

When $\mathcal P^+_{ik} < \mathcal P_i$ for all $k$, $\mathcal U_U$ is convex w.r.t. $\mathcal P_i$ over the interval  $[\max \mathcal P^+_{ik},\overline{\mathcal P}_i]$ using the same argument as above.  Similarly, when $\mathcal P_i < P^-_{ik}, \forall k \neq i$ $\mathcal U_U$ is convex w.r.t. $\mathcal P_i$ over the interval $[\underline{\mathcal P}_i,\min \mathcal P^-_{ik}]$. 
\end{proof}

\begin{lemma}
Suppose $\mathcal P^*_U$ is the optimal solution to (PU), we have $\mathcal P^*_U \in \tilde{\mathcal Q}$. 
\end{lemma}
\begin{proof} The result readily follows from Lemma \ref{lem:convex} and the fact that optimal solution over a closed interval for a convex function is an end point. 
\end{proof}
\begin{remark}
Since the linear approximation would perform better when $\frac{\mathcal P_i}{\mathcal N_0 + \sum_{j \in \bar{\Delta}^r_i \cap \bar{\Delta}^t_i} S_j(\mathbf{\mathcal P, \mathbf{p}})\cdot \mathcal P_{j}\cdot h_{ji}}$ is small, we know when the network size is large (i.e., $N$ is large) and dense, the upper bound problem will provide a good approximation for the original problem. 
\end{remark}

\begin{remark}
By finding the optimal solution for problems \textbf{(PL)} and \textbf{(PU)}, we have the bounds for the optimal solutions. 
\begin{align}
\mathcal U_L( \mathcal P^*_L) \leq \mathcal U(\mathcal P^*) \leq \mathcal U_U( \mathcal P^*_U)
\end{align}
Meanwhile we can use $\mathcal P^*_L$ and $\mathcal P^*_U$ as approximate strategies for our original problem \textbf{(P)} with 
\begin{align}
U(\mathcal P^*_L) \geq \mathcal U_L(\mathcal P^*_L), ~ U(\mathcal P^*_U)  \leq \mathcal U_U(\mathcal P^*_U)
\end{align} 
In next section we will focus on solving \textbf{(PL)} and \textbf{(PU)} instead of \textbf{(P)}.
\end{remark}

\subsection{Greedy search}

In solving \textbf{(PL)} and \textbf{(PU)} instead of the original \textbf{(P)}, the problem reduces to searching over a finite strategy space which can be done within a finite number of steps dependent on the size of the network. 
For the lower bound problem \textbf{(PL)}, the strategy space is on the order of $\mathcal O(N^N)$, while for the upper bound problem \textbf{(PU)} the order is $\mathcal O((2N-1)^N)$. However with a large scale WLAN network (referring to the number of APs in the network), these could still be excessively large even though finite. This is the classical \emph{rollout} problem in combinatorial optimization.  Below we present a heuristic greedy approach, which is shown later through numerical experiment to provide a near-optimal solution efficiently. The basic idea of a greedy search method is to maximize the system's total throughput w.r.t. a single variable at each stage of the computation while keeping the others fixed. 
The details of this approach is shown in Algorithm \ref{alg-greedy}, where the notation $\mathcal P_{-m}$ denotes the power profile for all but AP $m$, \rev{and the objective function ${\mathcal U}()$ denote either $\mathcal U_U$ or $\mathcal U_L$ depending on which problem (\textbf{(PU)} vs. \textbf{(PL)}) we are trying to solve.}  

\begin{algorithm}
\rule{5in}{1.5pt}
\caption{Pseudocode for Greedy Search }
\begin{algorithmic}
\STATE Set $\mathcal P_i(0) = \overline{\mathcal P}_i, \forall i \in \Omega$.\\
\STATE $\text{temp} = \mathcal U(\mathcal P(0))$, $\epsilon = \text{temp}$, $n=0$. \\
\While {$\epsilon > 0$}{ 
\STATE $n:=n+1$;
\STATE $m:=n \mod N$; 
\STATE $\mathcal P_m(n) = \argmax_{\mathcal P^m \in \tilde{\mathcal Q}_m} \mathcal U(\mathcal P^m, \mathcal P_{-m}(n-1))$; 
\FOR{$j = 1:N$}
\STATE if $j \neq m$ \\
\STATE ~~~$\mathcal P_j(n) = \mathcal P_j(n-1)$; \\
\ENDFOR
\STATE $\epsilon = |\mathcal U(\mathcal P(n)) - \text{temp}|$; \\
\STATE $\text{temp} = \mathcal U(\mathcal P(n))$;\\
}
\end{algorithmic}
\label{alg-greedy}
\rule{5in}{1.5pt}
\end{algorithm}

\begin{lemma}
The above greedy search terminates within a finite number of steps and reaches a local optimal solution to \textbf{(PL)} and \textbf{(PU)}. 
\end{lemma}

\subsection{Optimal search}
%

In this part, we present a randomized search algorithm that guarantees convergence to the optimal solution for \textbf{(PL)} and \textbf{(PU)}.
The algorithm works in rounds starting from AP $1$ and computes the power for one AP in each round.  Denote the state of the system at round $n$ as $\mathcal G(n)$ and $\mathcal G(n) := [\mathcal P_1(n),\mathcal P_2(n),...,\mathcal P_N(n)]$. Suppose at round $n$ AP $i$'s (i.e., $i = n \mod N$) power is being computed.  Then AP $i$'s next power level is updated using the following transition probability:  
\begin{align}
&\mathbb P (\mathcal G(n+1) = (\mathcal P^i,\mathcal G_{-i}(n))|\mathcal G(n)) \nonumber \\
&= \frac{e^{\mathcal U(\mathcal P^i, \mathcal G_{-i}(n))/\tau(n)}}{e^{\mathcal U(\mathcal P^i, \mathcal G_{-i}(n))/\tau(n)}+e^{\mathcal U(\mathcal G(n))/\tau(n)}}\cdot \frac{1}{\mathcal L_i}, \nonumber \\
&~\forall \mathcal P^i \in \tilde{\mathcal Q}_i, \mathcal P^i \neq \mathcal P_i(n).
\end{align}
with the probability of not changing the power level given as 
\begin{align}
&~~~~~~\mathbb P (\mathcal G(n+1) = \mathcal G(n)|\mathcal G(n)) \nonumber \\
&=\sum_{\mathcal P^i \in \tilde{\mathcal Q}_i, \mathcal P^i \neq \mathcal P_i(n)}\frac{e^{\mathcal U(\mathcal G(n))/\tau(n)}}{e^{\mathcal U(\mathcal P^i, \mathcal G_{-i}(n))/\tau(n)}+e^{\mathcal U(\mathcal G(n))/\tau(n)}}\cdot \frac{1}{\mathcal L_i}~. 
\end{align}
Here $\tau(n) := \frac{1}{n}$ is a positive smoothing factor and $\mathcal L_i$ is a normalization factor for user $i$. This search algorithm will be referred to as \textbf{P\_RAND}.
\rev{As before, the function ${\mathcal U}()$ in the above equation is the objective in problem \textbf{(PL)} (resp. \textbf{(PU)}) if the algorithm is used to search for an optimal solution to problem \text{(PL)} (resp. \text{(PU)}).} 

\begin{algorithm}
\rule{7in}{1.5pt}
\caption{Pseudocode for Randomized Search }
\begin{algorithmic}
\SetAlgoLined
\STATE Set $\mathcal P_i(0) = \overline{\mathcal P}_i, \forall i \in \Omega$.\\
\STATE $\text{temp} = \mathcal U(\mathcal P(0))$, $\epsilon = \text{temp}$, $n=0$. \\
\STATE $\text{thrs} = 10^{-4}$. (some small positive value)\\
\While {$\epsilon > \text{thrs}$}{ 
\STATE $n:=n+1$ ;
\STATE $m:=n \mod N$; 
\STATE Set $\mathcal P_m(n) = \mathcal P^m$  with probability $\frac{e^{\mathcal U(\mathcal P^m, \mathcal G_{-m}(n))/\tau(n)}}{e^{\mathcal U(\mathcal P^m, \mathcal G_{-m}(n))/\tau(n)}+e^{\mathcal U(\mathcal G(n))/\tau(n)}}\cdot \frac{1}{\mathcal L_m}, \forall \mathcal P^m \in \tilde{\mathcal Q}_m, \mathcal P^m \neq \mathcal P_m(n)$. 
\STATE While $\mathcal P_m(n)$ stays unchanged with probability $\sum_{\mathcal P^m \in \tilde{\mathcal Q}_i, \mathcal P^m \neq \mathcal  P_m(n)}\frac{e^{\mathcal U(\mathcal G(n))/\tau(n)}}{e^{\mathcal U(\mathcal P^m, \mathcal G_{-i}(n))/\tau(n)}+e^{\mathcal U(\mathcal G(n))/\tau(n)}}\cdot \frac{1}{\mathcal L_m}$.
\FOR{$j = 1:N$}
\STATE if $j \neq m$ \\
\STATE ~~~$\mathcal P_j(n) = \mathcal P_j(n-1)$; \\
\ENDFOR
\STATE $\epsilon = |\mathcal U(\mathcal P(n)) - \text{temp}|$; \\
\STATE $\text{temp} = \mathcal U(\mathcal P(n))$;\\
}
\end{algorithmic}
\label{alg:rand}
\rule{7in}{1.5pt}
\end{algorithm}

\begin{theorem}
\textbf{P\_RAND} converges to the optimal solution to the two approximate problems \textbf{(PL)} and \textbf{(PU)}.  
\end{theorem}
\begin{proof}
Due to the finiteness of the strategy spaces of all APs, we can form an $N$-dimensional positive recurrent Markov chain, \rev{with state at time $n$ given by $\mathcal G(n)$}; there exists a stationary distribution of this Markov chain.  

We next show that the following probability distribution $\pi$ over the state space $\tilde{\mathcal Q}$ is the stationary distribution of this Markov Chain.
\begin{align}
\pi(\mathcal G) = T\cdot e^{\mathcal U(\mathcal G)/\tau(n)}, \forall \mathcal G \in \tilde{\mathcal Q}~, 
\end{align}
where $T$ is the normalization constant. As $\sum_{\mathcal G \in \tilde{\mathcal Q}} \pi(\mathcal G) = 1$ the above is equivalent to 
$
T = \sum_{\mathcal G \in \tilde{\mathcal Q}}e^{\mathcal U(\mathcal G)/\tau(n)}~, 
$
and thus 
\begin{align}
\pi(\mathcal G) = \frac{e^{\mathcal U(\mathcal G)/\tau(n)}}{\sum_{\mathcal G \in \tilde{\mathcal Q}}e^{\mathcal U(\mathcal G)/\tau(n)}}, \forall \mathcal G \in \tilde{\mathcal Q}~,
\end{align}
To prove this, consider the \rev{detailed} balance equations. Specifically consider two states $\mathcal G_1$ and $\mathcal G_2$. Note that only when $\mathcal G_1$ and $\mathcal G_2$ differ in one element is the transition probability is positive; otherwise the transition probability is zero. We would like to show the following 
\begin{align}
\pi(\mathcal G_1)\cdot \mathbb P(\mathcal G_2|\mathcal G_1) = \pi(\mathcal G_2)\cdot \mathbb P(\mathcal G_1|\mathcal G_2),\forall \mathcal G_1,\mathcal G_2. \label{dbe}
\end{align}
When there is only one element difference between the two states we know
\begin{align}
\mathbb P(\mathcal G_2|\mathcal G_1) = \mathcal L\cdot \frac{e^{\mathcal U(\mathcal G_2/\tau(n))}}{e^{\mathcal U(\mathcal G_2/\tau(n))}+e^{\mathcal U(\mathcal G_1/\tau(n))}} \\
\mathbb P(\mathcal G_1|\mathcal G_2) = \mathcal L\cdot \frac{e^{\mathcal U(\mathcal G_1/\tau(n))}}{e^{\mathcal U(\mathcal G_1/\tau(n))}+e^{\mathcal U(\mathcal G_2/\tau(n))}}
 \end{align} 
 with $\mathcal L$ being some constant. Therefore (\ref{dbe}) holds readily. It follows from standard result that $\pi$ is indeed the stationary distribution. 

\rev{Denote $\mathcal G^* = \{\mathcal P^*\}$ as the set of global maximizers; i.e., the set of power profiles that maximizes our objectives in \textbf{(PL)} and \textbf{(PU)} respectively. Suppose there is a state $\mathcal G^{'} \notin \mathcal G^*$. We have 
\begin{align}
\pi(\mathcal G^{'}) =&  \frac{e^{\mathcal U(\mathcal G^{'})/\tau(n)}}{\sum_{\mathcal G \in \tilde{\mathcal Q}}e^{\mathcal U(\mathcal G)/\tau(n)}}= \frac{1}{\sum_{\mathcal G \in \tilde{\mathcal Q}}e^{(\mathcal U(\mathcal G)-\mathcal U(\mathcal G^{'}))/\tau(n)}} 
\end{align}
For $\mathcal G \in \mathcal G^*$ we have $\mathcal U(\mathcal G)-\mathcal U(\mathcal G^{'} )> 0$ and
\begin{align}
\lim_{n \rightarrow \infty}(\mathcal U(\mathcal G)-\mathcal U(\mathcal G^{'}))/\tau(n) \rightarrow \infty
\end{align}
and hence we have
$
\lim_{n \rightarrow \infty} \pi(\mathcal G^{'}) = 0
$
; furthermore by \cite{1987} we establish the following
\begin{align}
\lim_{n \rightarrow \infty} \pi(\mathcal G \in \mathcal G^*) = 1
\end{align}
i.e., the Markovian chain converges to the maximization states with probability 1. 
}
\end{proof}

\section{Solution Approach : A decentralized View} 
\subsection{A game theoretic view}
Starting from this section we analyze a strategic network for our problem. We first form a game theoretical model. Instead of group of users sharing a common goal, we assume the network users being strategic, i.e. with each one being rational. Define $\Delta \mathcal Q := \Delta \mathcal Q_1 \times \Delta \mathcal Q_2 \times ... \Delta \mathcal Q_N$. Here $\Delta \mathcal Q_i$ is the probability measure over action space $Q_i$. Then $(\Omega, \Delta \mathcal Q_i, \mathcal U_i)$ defines our game. We call this game (G1). We investigate the Nash Equilibrium (N.E.) of above game. To be specific, we show there is one unique N.E. exists for above game.
\begin{proposition}
The only N.E. for (G1) is $\mathbf{\overline{\mathcal P}} = [\overline{\mathcal P}_1,\overline{\mathcal P}_2,...,\overline{\mathcal P}_N]$.
\end{proposition}
\begin{proof} First we show $\mathbf{\overline{\mathcal P}}$ is a N.E.. Without of losing generality, consider AP $i$. As $\mathcal S_i(\mathbf{\mathcal P},\mathbf{p})$ is independent with $\mathcal P_i$ and $\mathcal C_i$ is a strictly increasing function w.r.t. $\mathcal P_i$, by unilaterally deviating from $\overline{\mathcal P}_i$ to some lower power profile $\mathcal P^{'}_i$, $\mathcal U_i$ will decrease. Therefore we proved $\mathbf{\overline{\mathcal P}}$ is a N.E.. From above arguments we also established the uniqueness immediately.

\end{proof}

\subsection{Mechanism design with perfect monitoring}
Starting from this part we solve our problem with mechanism design's approach. First of all since we are more interested in a dense network with well interacted APs we make the following assumption: for each user $i$, there is at least one $j$ such that
\begin{align}
\overline{\mathcal P}_j\cdot h_{ji} \geq \mathcal P^i_{cs}
\end{align}
i.e., for any AP it can be reached by at least one another AP within its carrier sensing range. Consider the case where there is such AP that nobody can reach him and let's call it a stand-alone AP. This AP can be thought of disconnected from the network and it cannot be controlled by the other APs' behaviors; or in a homogeneous network, this stand-alone AP will not contribute to other people's utility in a significant way and thus we could screen this case out and focus on the rest well connected APs.


We first start with the mechanism design problem with \emph{perfect monitoring}; here by \emph{perfect monitoring} we mean each user can monitor other users' transmission power in a perfect way without noise. As we discussed in last section, the socially optimal power allocation strategy profile may not be a NE of the multiuser system. Meanwhile, it is easy to check each user's utility function is a strictly increasing function w.r.t. the attempt rate $p$. Therefore in a decentralized system, each user has two incentives to deviate from our socially optimal strategy profile:
\begin{itemize}
\item[\RNum{1}.] Deviate from pre-specified transmission power to get a higher SNR.
\item[\RNum{2}.] Deviate from $p_c$ to get more air time.  
\end{itemize}
In this section we consider the optimization problem from a long run perspective for each user, i.e., we consider the following programming problem
$$
\begin{array}{ccll}
(P^*_i)&\max & \sum_{t=0}^{\infty} \delta^t\cdot \mathcal U_i(\mathbf{\mathcal P}(t), \mathbf{p}(t)) \\
&\st& \mathcal P_i(t) \in \mathcal Q_i, \forall i \in \Omega, t=0,1,2,...\end{array}
$$
Here $\delta$ is the discount factor satisfying $0 < \delta < 1$ and it models users or system's patience. Denote the social optimal power setting as $\mathbf{(\mathcal P^*,p)}$. We define the following states of APs:
\begin{myframe}
\begin{itemize}
\item[S.1.] $S_0$: The initial state; at $S_0$ the group of APs follow the strategy profile $\mathbf{(\mathcal P^*,p)}$. 
\item[S.2.] $S_i(t), t=1,2,...,L, i \in \Omega$: Punishment phases for user $i$. At state $S_i(t)$ APs follow the specified strategy profile  $\mathbf{(\overline{\mathcal P}},(\tilde{p}_i, [\tilde{p}_{ji}]_{j \neq i}))$. Here $L$ is some finite positive integer which is the length of punishment phases. 
\end{itemize}
\end{myframe}
Consider the following mechanism.
\begin{myframe}
\begin{itemize}
\item[M.1.] APs start at state $S_0$. If all APs follow the strategy profile $\mathbf{(\mathcal P^*,p)}$ specified for $S_0$, at next time APs will keep playing the strategy profile. 

\item[M.2.] If there is one AP, for example $i$ deviated from $\mathbf{(\mathcal P^*,p)}$ at time $t$, starting from time $t+1$, system goes into state $S_i(1)$ and play strategy profile $\mathbf{(\overline{\mathcal P}},(\tilde{p}_i, [\tilde{p}_{ji}]_{j \neq i})$. 
\item[M.3.] At state $S_i(t), t=1,2,...,L-1$, if all APs follow the specified strategy, system goes to state $S_i(t+1)$ at next time point; at $S_i(L)$, if all APs follow the specified strategy, system goes to state $S_0$. 

\item[M.4.] At state $S_i(t), t=1,2,...,L$, if AP $i$ deviates from the specified strategy profile, system goes to state $S_i(1)$ at next time. If AP $j \neq i$ deviates, system goes to $S_j(1)$. 
\end{itemize}
\end{myframe}
We refer this mechanism as (MPM). Now we present some requirements on choosing over $([\tilde{p}_i]_{i \in \Omega},[\tilde{p}_{j,k}]_{j \neq k})$.
\begin{myframe}
\begin{itemize}
\item[A.1.] $[\tilde{p}_{ji}]_{j \neq i}$ should be large enough to deter user $i$. Notice that with large enough $[\tilde{p}_{ji}]_{j \neq i}$ (with each elements close enough to 1) we have
\begin{align}
\mathcal S_i(\overline{\mathcal P}, \tilde{p}_i,[\tilde{p}_{ji}]_{j \neq i}) \rightarrow 0, \mathcal U_i(\overline{\mathcal P}, \tilde{p}_i,[\tilde{p}_{ji}]_{j \neq i}) \rightarrow 0
\end{align}
This essentially follows the intuition that when some user else becomes excessively aggressive, AP $i$ will lose most of its airtime (throttled by other users) and thus results a significant decrease of its achievable throughput. Therefore with appropriately chosen parameter set $\tilde{\mathbf{p}}$ we can get 
\begin{align}
\mathcal U_i\mathbf{(\overline{\mathcal P}},\tilde{p}_i,[\tilde{p}_{ji}]_{j \neq i})< \mathcal U_i\mathbf{(\mathcal P^*,p)}
\end{align}

\item[A.2.] A second restriction over selecting $([\tilde{p}_i]_{i \in \Omega},[\tilde{p}_{j,k}]_{j \neq k})$ is that each AP $i$ would like to stay at other users' punishment phase instead of its own, i.e., 
\begin{align}
\mathcal U_i\mathbf{(\overline{\mathcal P}},\tilde{p}_j,[\tilde{p}_{kj}]_{k \neq j}) \geq \mathcal U_i\mathbf{(\overline{\mathcal P}},\tilde{p}_i, [\tilde{p}_{ki}]_{k \neq i})
\end{align}
Through simple algebra we can show the existence of such strategy profiles w.r.t. APs' attempt rate. The details are thus omitted here.

\item[A.3.] A third restriction on selecting $([\tilde{p}_i]_{i \in \Omega},[\tilde{p}_{j,k}]_{j \neq k})$ we put is as following.
\begin{align}
\mathcal U_i(\overline{\mathcal P}, &  p_i = 1, [\tilde{p}_{ji}]_{j \neq i})-\mathcal U_i(\overline{\mathcal P}, \tilde{p}_i, [\tilde{p}_{ji}]_{j \neq i}) \nonumber \\
&< \mathcal U_i(\overline{\mathcal P},\mathbf{p}) - \mathcal U_i(\overline{\mathcal P}, \tilde{p}_i, [\tilde{p}_{ji}]_{j \neq i})
\end{align} 
i.e., $[\tilde{p}_i]_i$ is large enough. We will discuss the intuition later.  
\end{itemize}
\end{myframe}
\begin{remark}
It can be shown a trivial combination for  such $([\tilde{p}_i]_{i \in \Omega},[\tilde{p}_{j,k}]_{j \neq k})$ is setting each element to be 1, i.e., $([\tilde{p}_i]_{i \in \Omega},[\tilde{p}_{j,k}]_{j \neq k})  =\mathbf{ 1}$. Here $\mathbf{ 1}$ is the all-one vector with corresponding length. However we do not want these punishments to be too harsh and these three requirements help establish our mechanism while keep the punishments as light as possible.  
\end{remark}

Now we show with appropriately chosen parameter, (MPM) enforces the strategy profile $\boldsymbol{(\mathcal P^*,p)}$ for all APs.

\begin{theorem}
$\boldsymbol{(\mathcal P^*,p)}$ is enforceable under mechanism (MPM) with large enough $\delta$ and appropriately chosen $L$. 
\end{theorem}
\begin{proof} To prove the enforceability of $\boldsymbol{(\mathcal P^*,p)}$ we need to check no AP would like to go for a one step deviation at any state of (MPM) (according to the one step deviation principle ). 

First check state $S_0$. Consider an arbitrary AP $i$. Denote the utility from following the social optimal strategy profile for AP $i$ as $\mathcal U^*_i$. Then by following the specified strategy at $S_0$, the long run accumulated utility for AP $i$ is given by 
\begin{align}
\mathcal U_i(S_0) &= \delta^0\cdot \mathcal U^*_i  + ...+\delta^t\cdot \mathcal U^*_i+...= \frac{1}{1-\delta}\mathcal U^*_i
\end{align}
Next check the utility for AP $i$ by deviating to another profile $(\mathcal P^{'}_i, p^{'}_i)$. Remember due to the finiteness of $\mathcal Q_i$ and $p$ ($0 < p < 1$), we have some finite positive number $M_i$ such that
\begin{align}
\mathcal U_i((\mathcal P^{'}_i, p^{'}_i),\mathbf{(\mathcal P, p)}_{-i} ) \leq M_i
\end{align}
Also denote the utility of AP $i$ of adopting $\mathbf{(\overline{\mathcal P}},\tilde{p}_i, [\tilde{p}_{ji}]_{j \neq i})$ at punishment phase $S_i(t)$ as $\mathcal U^p_i$.  Then by deviation we know the aggregated utility for AP $i$ is 
\begin{align}
\mathcal U^{d}_i(S_0) \leq &M_i + \delta\cdot \mathcal U^p_i+ \delta^2 \cdot \mathcal U^p_i + ... \nonumber \\
&+ \delta^L \cdot \mathcal U^p_i + \delta^{L+1}\cdot \mathcal U^*_i + ...
\end{align}
Then we have
\begin{align}
\mathcal U_i(S_0) - \mathcal U^{d}_i(S_0) \geq ( \mathcal U^*_i-M_i) +(\mathcal U^*_i-\mathcal U^p_i )\cdot \frac{\delta-\delta^{L+1}}{1-\delta} 
\end{align}
As $\mathcal U^*_i-M_i \leq 0, \mathcal U^*_i-\mathcal U^p_i  > 0$, with a large enough $\delta$ and $L$ we could have 
\begin{align}
 \frac{\delta-\delta^{L+1}}{1-\delta}  \rightarrow L, ~(\mathcal U^*_i-\mathcal U^p_i) \cdot L &\geq M_i - \mathcal U^*_i
\end{align}
Denote a pair of suitable $(\delta, L)$ as $(\delta^1, L^1)$ and we proved with  $(\delta^1, L^1)$ there will be no profitable one step deviation for any AP at state $S_0$. 

Next we analyze the punishment phases. Consider $S_i(t)$. First notice that obviously all APs at any punishment phase with any stage $S_j(t), \forall j \in \Omega, t=1,2,...,L$ there is no incentive to deviate their power strategy. This follows from the observation each AP $i$'s utility $\mathcal U_i$ is strictly increasing w.r.t. $\mathcal P_i$. Thus no user would deviate from their maximum transmission power as used in the specified strategy profile at $S_j(t)$. Therefore we only need to consider APs deviating with attempt rate $\mathbf{p}$. 

First consider AP $i$. Denote $\mathcal U_i(S_i)$ as the utility AP $i$ can get by following $S_i$. By increasing its attempt rate, AP $i$ will increase its utility. But again due to the finiteness of $p_i$ we have a upper bound for this one step increase and we denote it as $M^{'}_i$. Then suppose at stage $t$, AP $i$'s utility by following specified strategy is given by 
\begin{align} 
\mathcal U_i(S_i(t)) = & \mathcal U_i(S_i) + \delta \cdot \mathcal U_i(S_i) + ... \nonumber \\
&+ \delta^{L-t-1}\mathcal U_i(S_i) + \delta^{L-t}\mathcal U_i(S_0) + ... 
\end{align}
and by deviating we have 
\begin{align}
\mathcal U^d_i(S_i(t)) \leq &M^{'}_i + \delta \cdot \mathcal U_i(S_i) + ... \nonumber \\
&+ \delta^{L}\mathcal U_i(S_i) + \delta^{L+1}\mathcal U_i(S_0) + ... 
\end{align}
Then take the difference we have
\begin{align}
&\mathcal U_i(S_i(t)) - \mathcal U^d_i(S_i(t)) \nonumber \\
&\geq \mathcal U_i(S_i) - M^{'}_i + \delta\cdot (\mathcal U_i(\mathcal P)-\mathcal U_i(S_i)) 
\end{align}
As $0 < M^{'}_i  - U_i(S_i)  < \mathcal U_i(\mathcal P)-\mathcal U_i(S_i)$ and with a large enough $\delta$ we have $\mathcal U_i(S_i) - M^{'}_i + \delta\cdot (\mathcal U_i(\mathcal P)-\mathcal U_i(S_i)) >0$. 

The last step is to check whether AP $j \neq i$ would deviate at state $S_i(t)$ or not. As we already put constraints on $([\tilde{p}_i]_{i \in \Omega},[\tilde{p}_{j,k}]_{j \neq k})$ such that each player would rather stay at other AP's punishment phase instead of their own, hereby AP $j$ would not deviate. 
\end{proof}

\subsection{Mechanism design with imperfect monitoring}
In this subsection we analyze the problem that each user has \emph{imperfect monitoring} over other users' transmission power at each decision period. In practice, the monitoring over other APs' deviation cannot be done in real time dues to multiple sources of noises, e.g., thermal noise. Therefore without a central monitor, each user needs to precisely detect a deviation by its own observations.

To be specific we consider the following noise model :  Instead of being $\mathcal P_j\cdot h_{ji}$, the received power at user $i$ (transmitted from user $j$) is given by
\begin{align}
\mathcal P_j\cdot h_{ji} + \mathcal N_{i}
\end{align} 
here $\mathcal N_i$ is a noise source at user $i$'s receiver side which follows a Gaussian distribution, $\mathcal N_i \sim \mathcal N(0,\sigma_i)$. Moreover we assume any pair $(\mathcal N_i, \mathcal N_j), i, j\in \Omega$ are correlated with correlation matrix $\Sigma_{ij}$:
\[ \Sigma_{ij} = \left[ \begin{array}{ccc}
\sigma_{ii} & \sigma_{ij}  \\
\sigma_{ji} & \sigma_{jj}  \end{array} \right].\] 

For each user we propose the following mechanism. Each user takes a threshold $\epsilon_i$ for detection. If $|\mathcal P_{ji} - \mathcal P_j\cdot h_{ji}| \leq \epsilon_i$, user $i$ will hold the punishment while taking the user $j$ being non-deviating.  On the other hand if $|\mathcal P_{ji} - \mathcal P_j\cdot h_{ji}| > \epsilon_i$, user $i$ will initialize the punishment phase for user $j$. Therefore we have
\begin{align}
\mathbb P^n_{ij} = \mathbb P(|\mathcal P_{ji} - \mathcal P_j\cdot h_{ji}| \leq \epsilon_i) = 2\Phi_i(\epsilon_i)-1\\
\mathbb P^p_{ij} = 1-\mathbb P(|\mathcal P_{ji} - \mathcal P_j\cdot h_{ji}| \leq \epsilon_i) = 2-2\Phi_i(\epsilon_i)
\end{align}
Here $\mathbb P^n_{ij}$ is the probability of detecting no deviation while $\mathbb P^p_{ij}$ is the probability of positive detection. We name this mechanism (MIM); and similarly with (MPM) we have enforceability as following. 

\begin{theorem}
$\boldsymbol{(\mathcal P^*,p)}$ is enforceable under mechanism (MIM) with large enough $\delta$, appropriately chosen $L$ and positive correlated noise sources.
\end{theorem}
\begin{proof}
To see this threshold based detection strategy is enforceable, we need to show:
\begin{itemize}
\item[\RNum{1}.] User $i$ would like to adopt the threshold based detection strategy.
\item[\RNum{2}.] User will not take advantage of this threshold based strategy of other users.   
\end{itemize}
For \RNum{1}, we can design a similar punishment phases and punishment strategies to deter users from deviating from the pre-specified strategy. 
\begin{lemma}
Users will not deviate from the pre-specified threshold strategy. 
\end{lemma}
\begin{proof}
Denote the event user $i$ detects a deviation as $E_{i}$ and the event for no detection triggered as $\bar{E}_{i}$. We make the following assumptions over $\Sigma_{ij}$ so that 
\begin{align}
\mathbb P(E_i|E_j) > \mathbb P(\bar E_i|E_j),\mathbb P(E_j|E_i) > \mathbb P(\bar E_j|E_i)
\end{align}
Or equivalently $
\mathbb P(E_i|E_j)  > \frac{1}{2}, \forall i,j \in \Omega$; and we call the noise sources are positively correlated with each other. We show an example as following to demonstrate what kinds of noise resources have positive correlation.

\begin{exam}. \emph{Bi-variate Gaussian Distribution}

Consider 
\[ \Sigma_{ij} = \left[ \begin{array}{ccc}
1 & \rho_{ij}  \\
\rho_{ij} & 1  \end{array} \right].\] 
The sufficient condition for a positive correlation is given as (algebraic details omitted for concise presentation)
\begin{align}
\Phi(\frac{\epsilon_i-\sqrt{\rho_{ij}}\epsilon_j}{\sqrt{1-\rho^2_{ij}}})+\Phi(\frac{\epsilon_i+\sqrt{\rho_{ij}}\epsilon_j}{\sqrt{1-\rho^2_{ij}}})> 1+\frac{1}{2}
\end{align}
Following which a more loose condition comes as 
\begin{align}
\Phi(\epsilon\cdot \frac{1-\sqrt{\rho_{ij}}}{\sqrt{1-\rho_{ij}^2}})> \frac{3}{4} \label{mimc}
\end{align}
Therefore following right after Equation (\ref{mimc}) a simple and sufficient condition for the positive correlation is 
\begin{align}
\epsilon_i = \epsilon_j, \rho_{ij} \rightarrow 1
\end{align}
and $\epsilon_i, i \in \Omega$ are properly chosen. The results show that in a network with homogeneous noise sources we need the correlation factor to be high enough. This follows the intuition of designing deviation-proof mechanisms for private imperfect monitoring problems: only when the correlation of private observations are highly positive correlated.  
\end{exam}

Therefore when a user deviates, the probability of being detected by other users becomes higher than sticking with prescribed strategy. As long as we have a harsh and long enough punishment phase for each users, nobody will deviate. The design details follow similar path with (MIM) as in the perfect monitoring section and thus omitted. 

\end{proof}
Now we consider \RNum{2} and particularly we would like to see nobody could be able to take use of the ``cushion'' tolerance of other users by the threshold policy. 
\begin{lemma}
No user will deviate to take use of the threshold detection strategy of other users.
\end{lemma}
\begin{proof}
Without losing generality, consider user $i$. Suppose user $i$ increase $\mathcal P^*_i$ to $\mathcal P^*_i+\tau$. Denote the benefits of the deviation for user $i$ as $\Delta~\mathcal U_i(\tau)$. Firstly with the change, at user $j \neq i$'s side, the probability of punishment initialized becomes
\begin{align}
\mathbb P^{p,\tau}_{ji}& = 1-\mathbb P(-\epsilon_j \leq \mathcal N_j+\tau\cdot h_{ij} \leq \epsilon_j) \nonumber \\
&= 2-(\Phi_j(\epsilon_j+\tau\cdot h_{ij})+\Phi_j(\epsilon_j-\tau\cdot h_{ij}))
\end{align}
By basic algebra we can show (details omitted)
\begin{align}
2-(\Phi_j(\epsilon_j+\tau\cdot h_{ij})&+\Phi_j(\epsilon_j-\tau\cdot h_{ij}))\nonumber \\
& > 2-2\cdot\Phi_j(\epsilon_j)
\end{align}
i.e., $\mathbb P^{p,\tau}_{ji} > \mathbb P^{p}_{ji}$; and furthermore we have
$
\frac{\partial (\mathbb P^{p,\tau}_{ji} - \mathbb P^{p}_{ji})}{\partial \epsilon_j} <0
$. 
The intuition here is when $\epsilon_j$ is smaller, user $i$ would be detected with transmit power change $\tau$ more easily. Therefore by designing appropriate detection thresholds $\epsilon_j, j\in \Omega$ and punishment phases with large enough $\delta$ and $L$ we have the following
\begin{align}
(\mathbb P^{p,\tau} - \mathbb P^{p})\{&\underbrace{(\mathcal U^*_i-\mathcal U^p_i)+\delta (\mathcal U^*_i-\mathcal U^p_i)+...+\delta^{L-1}(\mathcal U^*_i-\mathcal U^p_i)}_{L}\}\nonumber \\
&> \Delta~\mathcal U_i(\epsilon)
\end{align}
Therefore each user would be deterred from deviating under this threshold enabled detection mechanism. 
\end{proof}

\end{proof}

\begin{theorem}
The (MIM) help increases system performance compared with (MPM) under a imperfect monitoring system.
\end{theorem}
\begin{proof} 
We sketch the basic idea here. Obviously by introducing in the $\epsilon$ tolerance threshold we decreased the false alarm probability. As the network performance degradation caused by punishment is more severe than the degradation by the other users' slight deviation (light deviation under tolerance) or by the noise, we know the (MIM) helps improve system performance. 
\end{proof}

\section{Numerical Experiments}\label{sec:simulate}

In this section, we provide simulation results to show system performance under the greedy and randomized search algorithms (denoted as ``Greedy'' and ``\textbf{P\_RAND}'' in the figures, respectively).  We further compare them with the maximum transmission power strategy (``Max''), \textbf{PPHY} and \textbf{PMAC} respectively.  The WLAN network's topology used in the experiment is randomly generated, with 10 APs placed according to a uniform distribution in a square area; this topology is shown in Fig.\ref{top_2}. 
\begin{figure}[h!]
\centering
\includegraphics[width=0.5\textwidth,height=0.3\textwidth]{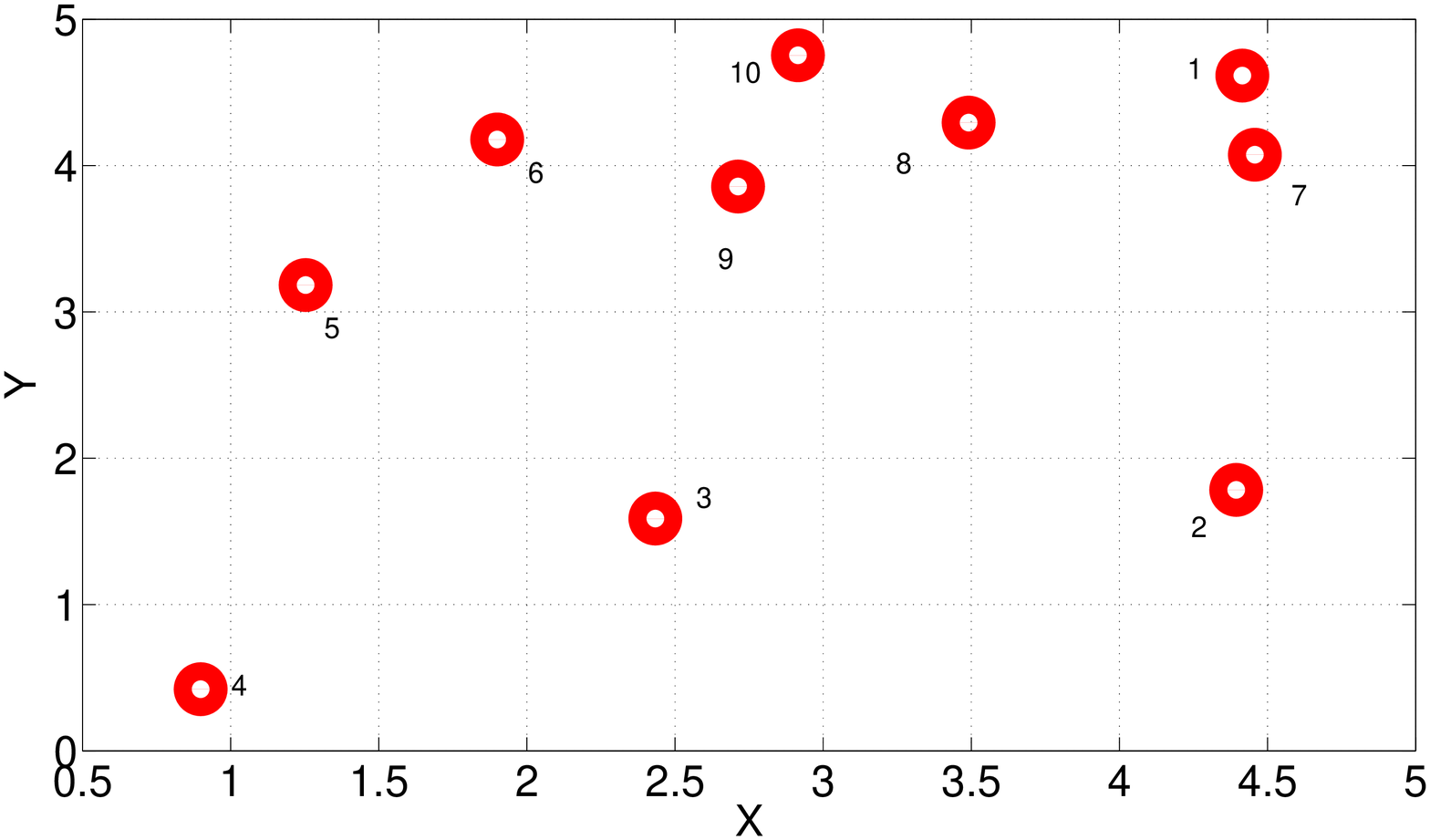}
\caption{Topology of the AP Network
\com{I suggest we add labels/indices to all the APs...}}\label{top_2}
\end{figure}

\subsection{Optimization with dual effects}
 
We begin by comparing the computed power levels and the resulting system-wide throughput under the greedy and randomized search algorithms and the fixed, maximum power scheme. 
\begin{figure}[h!]
\centering
\includegraphics[width=0.5\textwidth,height=0.2\textwidth]{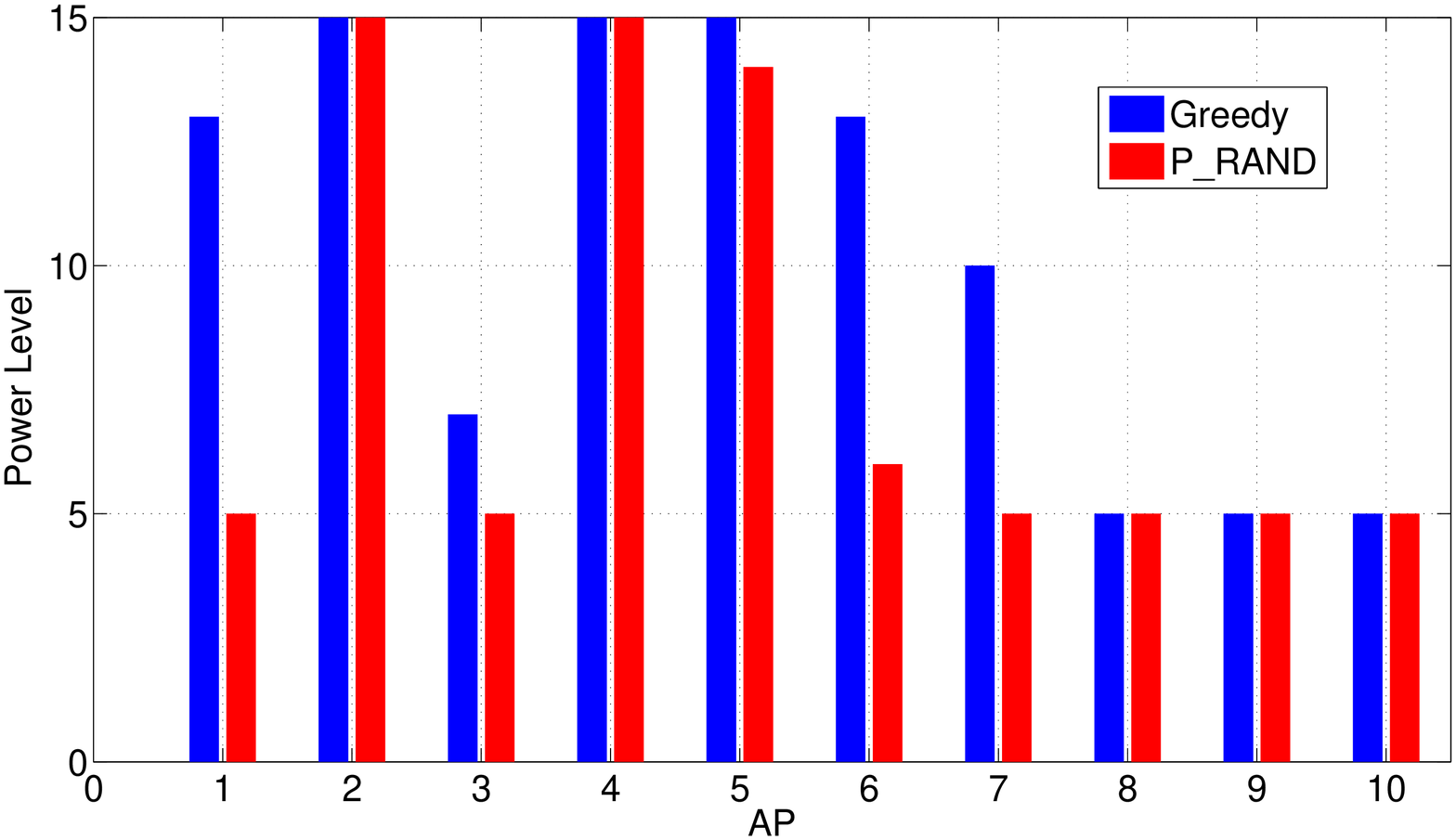}
\caption{Output power levels under Greedy Search \& \textbf{P\_RAND}.} 
\label{sim1_power}
\end{figure}

\begin{figure}[h!]
\centering
\includegraphics[width=0.5\textwidth,height=0.2\textwidth]{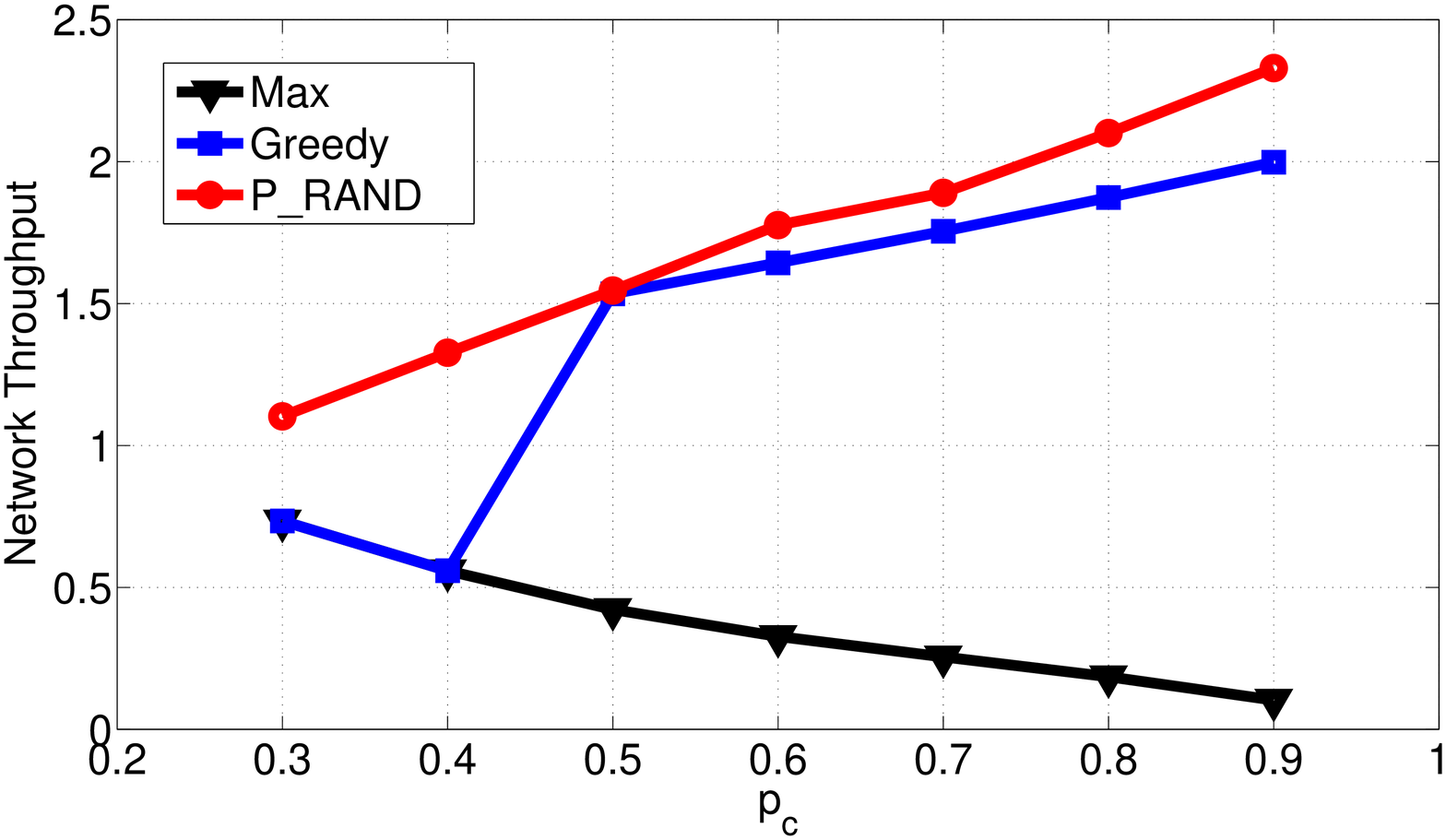}
\caption{Comparison of system-wide throughput.}  \label{thr_1}
\end{figure}


\rev{For the first set of results, we fix $p_c = 0.6$ and a maximum transmission power level of 15 for all APs.  The resulting optimal power profile is depicted in Fig.\ref{sim1_power}. Here to get the ``optimal solution'' we utilize \textbf{P\_RAND} to solve \textbf{(PL)} and \textbf{(PU)} separately and then choose the one that gives us a better total throughput. }  \com{Here by ``optimal'' do we refer to the optimal solution to (PU) or (PL)?} 
We see that in this case APs $(8,9,10)$'s power levels are far short of the maximum level. This reflects the need to avoid 
excessive interference with each other as they are clustered in a relatively crowded neighborhood.  \rev{APs $(2,4)$ are sitting relatively ``alone'' and thus they could transmit at a higher power. Similar observations can be made at each AP. } 
\com{With complete labeling on the topology we should be able to complete this paragraph nicely...} 

Next, the system performance is shown in Fig. \ref{thr_1} as a function of the attempt rate $p_c$.  It is interesting to note the opposite trends exhibited by using optimal power tuning vs. always using maximum power levels as the attempt rate increases.  As the network gets busier (more congested with higher attempt rate), the maximum power levels exacerbates the problem and the system throughput degrades even though the APs are trying harder.  On the other hand, using optimal power-tune, as the network becomes more congested, the APs react by decreasing their transmission powers appropriately so that the system throughput actually improves. 
By either the greedy search or the randomized search algorithm, our optimal power-tune problem helps achieve a significant throughput performance improvement compared to the static maximum transmission power scheme. 

We end this part with a look into the convergence performance of \textbf{P\_RAND}, shown in Fig.\ref{conv_3}. It is seen that our randomized search algorithm converges quickly to the end solution; under the same simulation setting, the greedy policy converges to a solution of system throughput at around 1.6. 
\begin{figure}[h!]
\centering
\includegraphics[width=0.5\textwidth,height=0.2\textwidth]{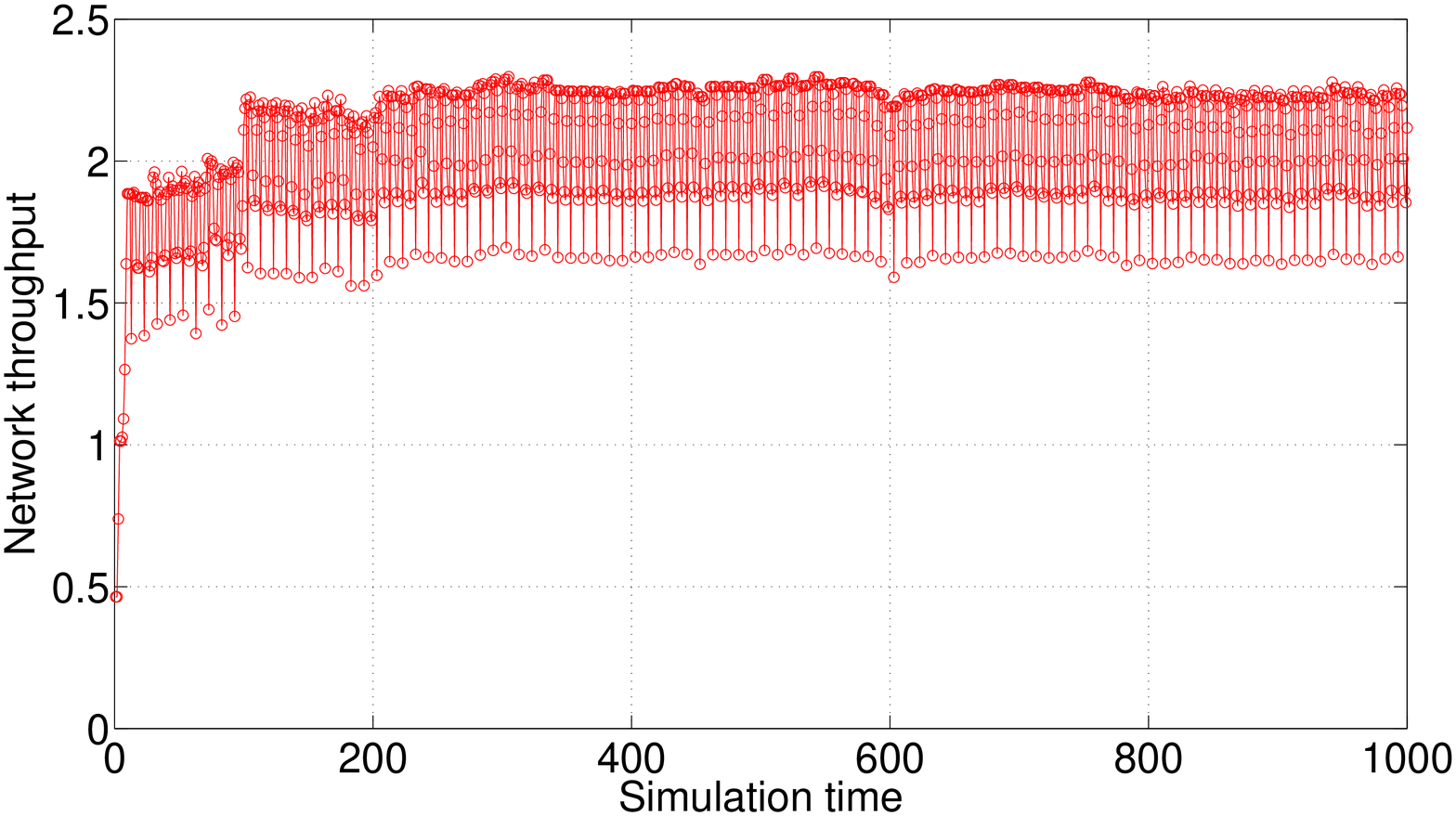}
\caption{Convergence of \textbf{P\_RAND}\com{suggest changing the y-label to ``network throughput'' just to be consistent}}\label{conv_3}
\end{figure}

\subsection{Compare with \textbf{PPHY}}
We next compare our optimization model with dual effect to the model given by \textbf{PPHY}, which tries to maximize the total rate at the physical layer without considering contention. 
The achievable throughput at each AP node (\rev{under attempt rate $p_c = 0.6$} ) is shown in Fig. \ref{PHY2} while the transmission power returned by \textbf{PPHY} vs. that by \textbf{P\_RAND} is shown in Fig. \ref{PHY1}. \com{at what $p_c$ level?} 
 \begin{figure}[h!]
\centering
\includegraphics[width=0.5\textwidth,height=0.2\textwidth]{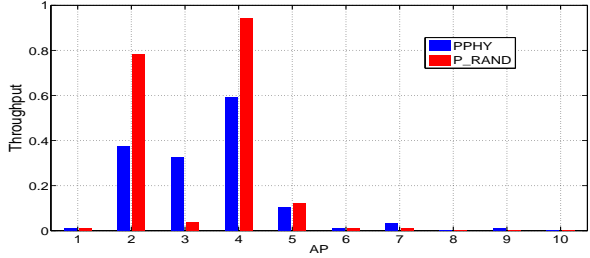}
\caption{AP's individual throughput}\label{PHY2}
\end{figure}
 \begin{figure}[h!]
\centering
\includegraphics[width=0.5\textwidth,height=0.2\textwidth]{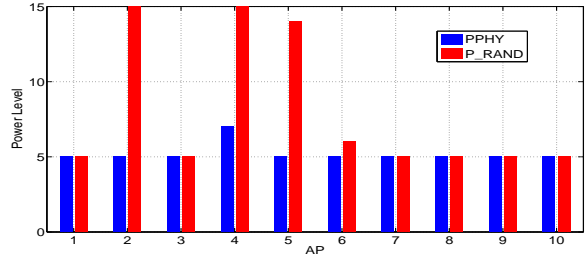}
\caption{power-tune results}\label{PHY1}
\end{figure}

We see a clear difference in how power levels are tuned and the resulting throughput across different AP nodes. 
The reason is that under \textbf{PPHY} each AP treats all other APs as noise resources. However, due to CSMA/CA, no parallel transmission would be allowed for APs within the carrier sensing range and thus the first-order noises (those from the closest neighbors) could be removed. Therefore APs could increase their power to some extent without contributing too much to their neighbors' noise level.  This is why we observe a few APs with much higher power under \textbf{P\_RAND} than under \textbf{PPHY}.  By contrast, with only PHY layer  optimization APs cannot take full advantages of the noise-free property of CSMA, and therefore act  conservatively.

To make our comparison complete we present the total system throughput performance in Fig. \ref{PHY3}.  We see \textbf{PPHY} is clearly out-performed by our optimization model especially under higher attempt rate. 
 \begin{figure}[h!]
\centering
\includegraphics[width=0.5\textwidth,height=0.2\textwidth]{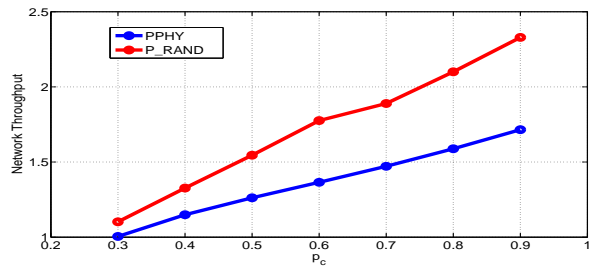}
\caption{Comparison of network's total throughput}\label{PHY3}
\end{figure}

\subsection{Compare with \textbf{PMAC}}
We perform a similar comparison with \textbf{PMAC}.  We start with a comparison of overall contention order under different SNR constraints in Fig.\ref{MAC1} (\rev{under attempt rate $p_c = 0.6$, same for Fig.\ref{MAC2})}. With a higher base SNR, the required transmission power is potentially higher under \textbf{PMAC}, and thus the total contention order increases.  The contention order under \textbf{P\_RAND} on the other hand stays constant. 
\begin{figure}[h!]
\centering
\includegraphics[width=0.5\textwidth,height=0.2\textwidth]{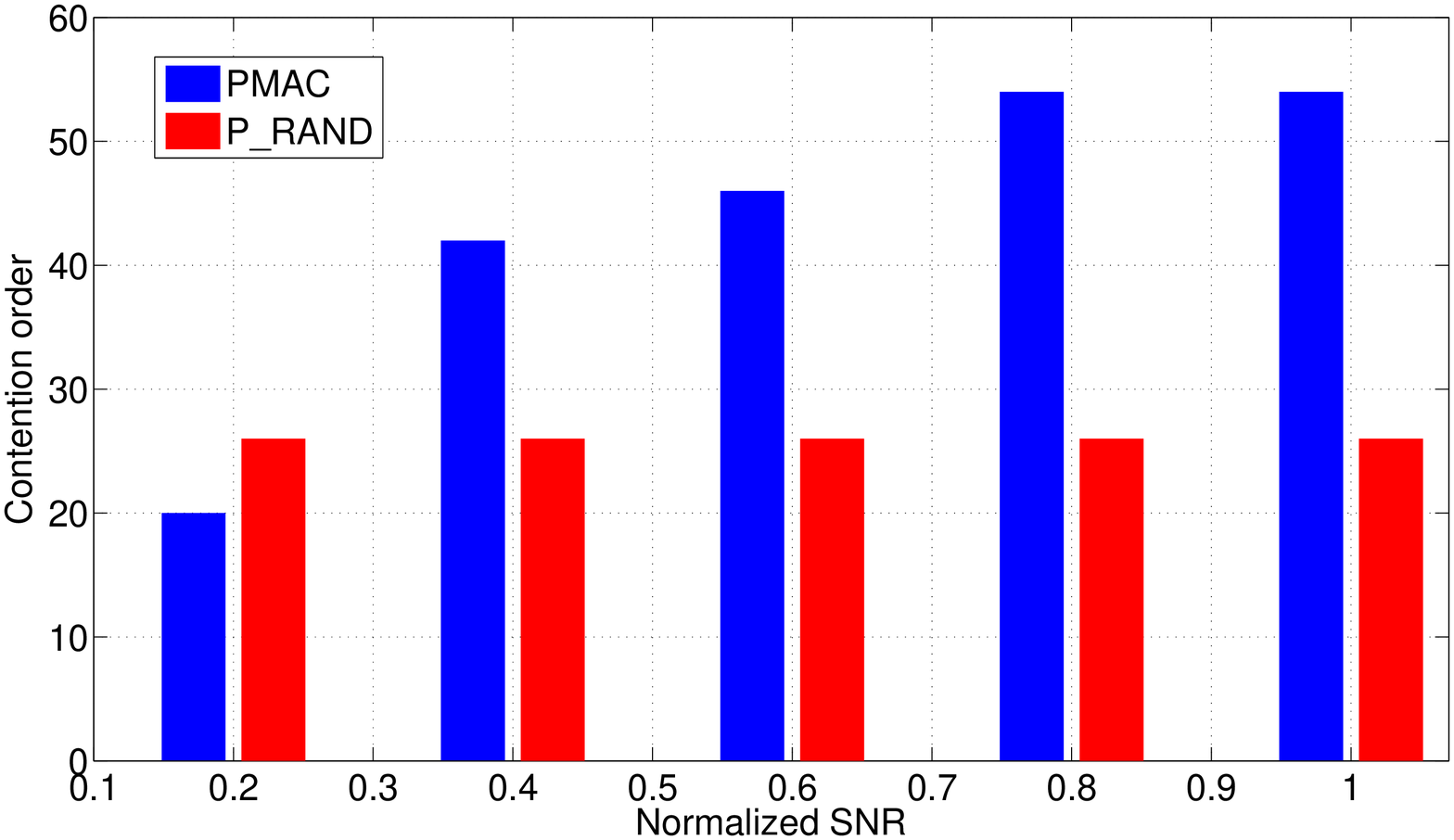}
\caption{Contention order comparison}\label{MAC1}
\end{figure}

\com{Shouldn't we also show an individual power level result as in the previous subsection?}
\response{yang: here i'm using the contention order one since 1. this is of particular interests from this section; 2. it to some extent reflects the power tuning; 3. space limitation..and it will look quite similar with results in previous sub-sections.}
We see that in order to reduce the contention order the APs again act rather conservatively in reducing their power levels.  This leads to a drop in noise resistance and the overall network throughput, as shown in 
Fig.\ref{MAC2} and Fig. \ref{MAC3}, respectively.  
\com{These figures are at what $p_c$ level?} 

\begin{figure}[h!]
\centering
\includegraphics[width=0.5\textwidth,height=0.2\textwidth]{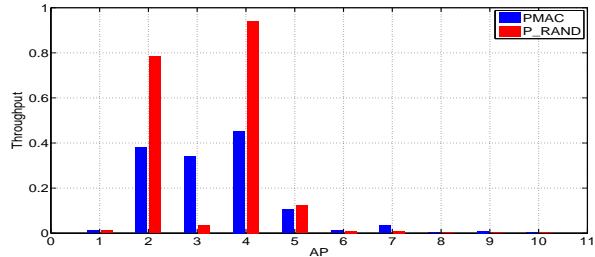}
\caption{APs' individual throughput}\label{MAC2}
\end{figure}
\begin{figure}[h!]
\centering
\includegraphics[width=0.5\textwidth,height=0.2\textwidth]{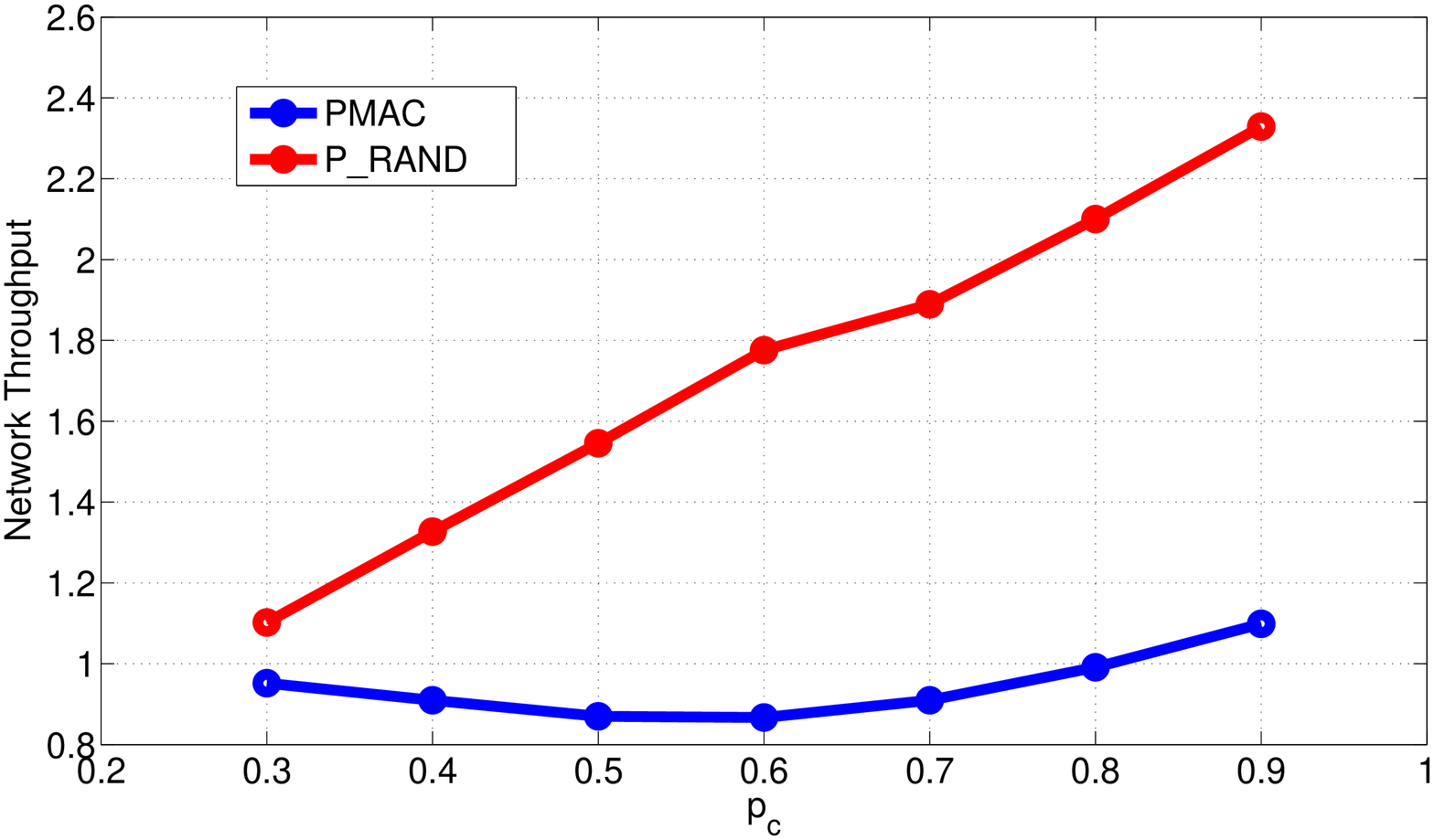}
\caption{Comparison of network's total throughput}\label{MAC3}
\end{figure}
%
%
\section{Related works}\label{sec:related}

There have been many classical PHY layer power-tune studies using Shannon's capacity formula. For example, Kim et al. investigated a transmit power and carrier sensing threshold tuning problem for improving spatial reuse in \cite{Kim:2006:ISR:1161089.1161131}. Chiang et al. looked into the transmit power control problem through  management of interference, energy and connectivity in \cite{Chiang:2008:PCW:1454705.1454706}. 
In \cite{Phan:2009:CDP:1817770.1818091}, Phan et al. investigate distributed power control problem on physical layer; a distributed algorithm is given and critical performance criteria, such as convergence are analyzed. In \cite{5701695}, Tan et al. analyze several multi-user spectrum management problems with focus on power control. 

More recently, power control problems have been analyzed under game theoretical framework. Sharma et al. proposed a game theoretical approach for decentralized power allocation in \cite{Sharma:2008:GAD:1536956.1536958}. In \cite{DBLP:journals/ton/TekinLSHA12}, a congestion game model is proposed to analyze power control problem as a form of resource allocation. Equilibrium strategies have been given under certain assumptions. In \cite{DBLP:journals/corr/abs-1111-2456}, a power control problem is modeled as repeated games with strategic users and intervention theory is proposed to induce target strategy from users. Imperfect monitoring repeated game model is analyzed in \cite{DBLP:journals/corr/abs-1201-3328} with the assumption of a Local Spectrum Server (LSS). In \cite{DBLP:conf/infocom/WanCDWY12}, Wan et al. consider a power control problem w.r.t. reducing contention order on the link layer while keeping the physical layer interference under certain levels. 

In terms of computation, for standard integer optimization (or combinatorial optimization) problems researchers typically seek relaxation to convert the problem into a continuous problem in the hope it can be solved by standard LP or convex algorithms; in \cite{Young98,Kirkpatrick83optimizationby,10.1109/ICNP.2007.4375832}, efficient search algorithms have been proposed to tackle finite space optimization problems. 

\section{Conclusion}\label{sec:conclusion}
With the proliferation of densely-deployed WLANs, power tuning becomes a critical problem as it has major impacts on SNR as well as contention levels in these networks. Prior works mostly focused on one of these two issues in pursuit of either higher throughput or lower contention level, but not both. 

In this paper, we have investigated the network throughput optimization problem by optimizing both spatial reuse (MAC) and SNR (PHY) performance at the same time. We have presented the complexity of solving the joint optimization problem and derived approximations to make it tractable. Then, by analyzing the problem structure, we have proposed efficient and near-optimal solutions. In order to demonstrate the effectiveness of our approach, we compared our results with several models optimizing on only either PHY or MAC layer. A clear advantage has been demonstrated for the cross-layer approach.

In the second part, we first show that the N.E. for the decentralized WLAN appears to be inefficient; then we show specific punishment mechanism can be designed to enforce the social near-optimal solution with our system under both perfect and imperfect monitoring environment.

\section*{Acknowledgment}
This work was partly performed while the first author interned with Juniper Networks, Inc. We would like to extend our appreciations to David Aragon, Joe Williams, and many others for the inspiring discussions and helpful comments.
 
\bibliography{myref}

\begin{thebibliography}{10}

\bibitem{1987}
Shoshana Anily and Awi Federgruen.
\newblock Ergodicity in parametric nonstationary markov chains: An application
  to simulated annealing methods.
\newblock {\em Operations Research}, 35(6):pp. 867--874, 1987.

\bibitem{Chiang:2008:PCW:1454705.1454706}
Mung Chiang, Prashanth Hande, Tian Lan, and Chee~Wei Tan.
\newblock Power control in wireless cellular networks.
\newblock {\em Found. Trends Netw.}, 2(4):381--533, April 2008.

\bibitem{Kim:2006:ISR:1161089.1161131}
Tae-Suk Kim, Hyuk Lim, and Jennifer~C. Hou.
\newblock Improving spatial reuse through tuning transmit power, carrier sense
  threshold, and data rate in multihop wireless networks.
\newblock In {\em Proceedings of the 12th annual international conference on
  Mobile computing and networking}, MobiCom '06, pages 366--377, New York, NY,
  USA, 2006. ACM.

\bibitem{Kirkpatrick83optimizationby}
S.~Kirkpatrick, C.~D. Gelatt, and M.~P. Vecchi.
\newblock Optimization by simulated annealing.
\newblock {\em Science}, 220:671--680, 1983.

\bibitem{Phan:2009:CDP:1817770.1818091}
Khoa~T. Phan, Long~Bao Le, Sergiy~A. Vorobyov, and Tho Le-Ngoc.
\newblock Centralized and distributed power allocation in multi-user wireless
  relay networks.
\newblock In {\em Proceedings of the 2009 IEEE international conference on
  Communications}, ICC'09, pages 4396--4400, Piscataway, NJ, USA, 2009. IEEE
  Press.

\bibitem{Sharma:2008:GAD:1536956.1536958}
Shrutivandana Sharma and Demosthenis Teneketzis.
\newblock A game-theoretic approach to decentralized optimal power allocation
  for cellular networks.
\newblock In {\em Proceedings of the 3rd International Conference on
  Performance Evaluation Methodologies and Tools}, ValueTools '08, pages
  1:1--1:10, ICST, Brussels, Belgium, Belgium, 2008. ICST (Institute for
  Computer Sciences, Social-Informatics and Telecommunications Engineering).

\bibitem{10.1109/ICNP.2007.4375832}
Yang Song, Chi Zhang, and Yuguang Fang.
\newblock Throughput maximization in multi-channel wireless mesh access
  networks.
\newblock {\em 2012 20th IEEE International Conference on Network Protocols
  (ICNP)}, 0:11--20, 2007.

\bibitem{DBLP:conf/globecom/TanPC05}
Chee~Wei Tan, Daniel~P. Palomar, and Mung Chiang.
\newblock Solving nonconvex power control problems in wireless networks: low
  sir regime and distributed algorithms.
\newblock In {\em GLOBECOM}, page~6, 2005.

\bibitem{DBLP:journals/ton/TekinLSHA12}
Cem Tekin, Mingyan Liu, Richard Southwell, Jianwei Huang, and Sahand Haji~Ali
  Ahmad.
\newblock Atomic congestion games on graphs and their applications in
  networking.
\newblock {\em IEEE/ACM Trans. Netw.}, 20(5):1541--1552, 2012.

\bibitem{DBLP:conf/infocom/WanCDWY12}
Peng-Jun Wan, Dechang Chen, Guojun Dai, Zhu Wang, and F.~Frances Yao.
\newblock Maximizing capacity with power control under physical interference
  model in duplex mode.
\newblock In {\em INFOCOM}, pages 415--423, 2012.

\bibitem{5701695}
Chee wei Tan, S.~Friedland, and S.H. Low.
\newblock Spectrum management in multiuser cognitive wireless networks:
  Optimality and algorithm.
\newblock {\em Selected Areas in Communications, IEEE Journal on},
  29(2):421--430, 2011.

\bibitem{DBLP:journals/corr/abs-1111-2456}
Yuanzhang Xiao, Jaeok Park, and Mihaela van~der Schaar.
\newblock Repeated games with intervention: Theory and applications in
  communications.
\newblock {\em CoRR}, abs/1111.2456, 2011.

\bibitem{DBLP:journals/corr/abs-1201-3328}
Yuanzhang Xiao and Mihaela van~der Schaar.
\newblock Dynamic spectrum sharing among repeatedly interacting selfish users
  with imperfect monitoring.
\newblock {\em CoRR}, abs/1201.3328, 2012.

\bibitem{Young98}
H.~Peyton Young.
\newblock {\em Individual Strategy and Social Structure}.
\newblock Princeton University Press, 1998.

\end{thebibliography}
\end{document}